\documentclass[envcountsame]{llncs}

\usepackage{amsmath}
\usepackage{amsfonts}
\usepackage{amssymb}
\usepackage{latexsym}
\usepackage{stmaryrd}
\usepackage{array}
\usepackage{exscale}
\newcommand{\nc}{\newcommand}
\newcommand{\ol}{\overline}
\newcommand{\ul}{\underline}
\newcommand{\es}{\emptyset}
\newcommand{\sm}{\setminus}
\newcommand{\ve}{\varepsilon}
\newcommand{\vp}{\varphi}

\newcommand{\ra}{\rightarrow}

\newcommand{\sse}{\subseteq}

\newcommand{\fa}{\forall}

\newcommand{\mr}{\mathrm}
\newcommand{\mc}{\mathcal}
\newcommand{\mf}{\mathfrak}

\newcommand{\DMO}{\DeclareMathOperator}
\newcommand{\DST}{\displaystyle}

\newcommand{\ZZ}{\mathbb{Z}}
\newcommand{\NN}{\mathbb{N}}
\newcommand{\NNZ}{\NN_0}

\newcommand{\RR}{\mathbb{R}}

\newcommand{\PP}{\mathbb{P}}

\usepackage{listings}
\newcommand{\inl}[1]{\lstinline$#1$}
\newcommand{\set}[1]{\{ #1 \}}
\newcommand{\setb}[1]{\big \{ \, #1 \, \big \}}
\DMO{\dom}{dom}
\DMO{\id}{id}
\DMO{\cod}{cod} 
\DMO{\rg}{rg} 
\DMO{\tcomp}{\trans{\circ}} 
\DMO{\simrv}{\,\sim\hspace{-0.05em}}
\DMO{\simlv}{\!\sim\,}
\nc{\simlvi}[1]{\!\sim_{#1}}
\DMO{\card}{card}
\DMO{\proj}{pr}
\DMO{\inj}{in}
\DMO{\symdif}{\vartriangle} 
\DMO{\addcup}{{\stackrel{\text{\raisebox{-2.2ex}[-0ex][-0ex]{\large$\cdot$}}}{\cup}}} 
\DMO{\addbcup}{{\stackrel{\text{\raisebox{-4.2ex}[-0ex][-0ex]{\Large$\cdot$}}}{\bigcup}}} 
\nc{\apprel}[3]{{#1}(#2)_{(#3)}} 
\DMO{\Rel}{\mf{REL}} 
\DMO{\Abb}{\mf{MAP}} 
\DMO{\Tra}{\mf{T}} 
\DMO{\Per}{\mf{S}} 
\DMO{\Pert}{\Per_t} 
\DMO{\Ptr}{\mf{PT}} 
\DMO{\fix}{fix} 
\DMO{\Peri}{\Per_i} 
\DMO{\Pers}{\Per_s} 
\DMO{\Rrel}{\Rel_r} 
\DMO{\Srel}{\Rel_s} 
\DMO{\Trel}{\Rel_t} 
\DMO{\konkat}{\sqcup} 
\DMO{\cmpl}{\complement^1} 
\nc{\cmpli}[1]{\complement^1_{#1}} 
\DMO{\cmplz}{\complement^0} 
\nc{\cmplzi}[1]{\complement^0_{#1}} 
\DMO{\cmplzo}{\complement^*} 
\nc{\cmplzoi}[1]{\complement^*_{#1}} 
\DMO{\fsigma}{{\mf{F}}_{\sigma}} 
\DMO{\gdeltao}{\mf{G}_{\sigma}}
\DMO{\fs}{{\mf{F}}_{s}} 
\DMO{\fss}{{\mf{F}}_{s}^*} 
\nc{\zf}{\mr{ZF}}
\nc{\zfmf}{\zf^0} 
\nc{\zfc}{\mr{ZFC}}
\nc{\zfcmf}{\zfc^0} 
\nc{\bst}{\mr{BST}} 
\newcommand{\tb}[2]{\set{#1, \dots, #2}} 
\DMO{\re}{Re}
\DMO{\im}{Im}
\DMO{\sgn}{sgn} 
\providecommand{\abs}[1]{\lvert #1 \rvert} 
\providecommand{\ioa}[2]{\,]#1, #2]}
\nc{\untit}[2]{{#1}^{#2 \downarrow}} 
\nc{\obit}[2]{{#1}^{#2 \uparrow}} 
\DMO{\inttop}{\tau_{\mr{O}}} 
\DMO{\rointtop}{\tau_+} 
\DMO{\lointtop}{\tau_-} 
\DMO{\sid}{\mc{IDL}} 
\DMO{\skid}{\mc{CID}} 
\DMO{\smid}{\sid_m} 
\DMO{\smkid}{\skid_m} 
\DMO{\nachbarn}{\Gamma}
\DMO{\enachbarn}{N}
\DMO{\nachbarnr}{\Gamma_{\!\mr{r}}}
\DMO{\nachbarnz}{\widetilde{\Gamma}}
\DMO{\nachbarnzr}{\widetilde{\Gamma}_{\!\mr{r}}}
\DMO{\inzEK}{\mc{I}^{\mr{V}}}
\DMO{\inzEKe}{\mc{I}^{\mr{V}}_1}
\DMO{\inzEKz}{\mc{I}^{\mr{V}}_2}
\nc{\inzEKi}[1]{\mc{I}^{\mr{V}}_{#1}}
\DMO{\inzKE}{\mc{I}^{\mr{E}}}
\DMO{\inzKEe}{\mc{I}^{\mr{E}}_1}
\DMO{\inzKEz}{\mc{I}^{\mr{E}}_2}
\nc{\inzKEi}[1]{\mc{I}^{\mr{E}}_{#1}}
\DMO{\inz}{I}
\DMO{\tinz}{\trans{\inz}} 
\DMO{\adjE}{\mc{A}^{\mr{V}}}
\DMO{\adjEe}{\mc{A}^{\mr{V}}_1}
\DMO{\adjEz}{\mc{A}^{\mr{V}}_2}
\nc{\adjEi}[1]{\mc{A}^{\mr{V}}_{#1}}
\DMO{\adjor}{\mc{A}_{\mr{S}}} 
\DMO{\adjK}{\mc{A}^{\mr{E}}}
\DMO{\adj}{A}
\DMO{\degmin}{\ul{\deg}}
\DMO{\degmax}{\ol{\deg}}
\DMO{\degdur}{\widetilde{\deg}} 
\DMO{\ideg}{idg} 
\DMO{\odeg}{odg} 
\DMO{\degmaxl}{\ol{\deg}_{<}}
\DMO{\degl}{\deg_{<}}
\DMO{\rankmin}{\ul{\rank}}
\DMO{\rankmax}{\ol{\rank}}
\DMO{\rankdur}{\widetilde{\rank}}
\DMO{\rankmaxl}{\ol{\rank}_{<}}
\DMO{\rankl}{\rank_{<}}
\DMO{\vertexcon}{\kappa} 
\DMO{\edgecon}{\lambda} 
\DMO{\treewidth}{tw} 
\DMO{\girth}{g} 
\DMO{\circumference}{cf} 
\DMO{\length}{lgth} 
\DMO{\npm}{\Phi} 
\DMO{\concomp}{cc} 
\DMO{\nconcomp}{ncc} 
\DMO{\indprim}{ip} 
\DMO{\indimprim}{iip} 
\DMO{\bouquet}{B}
\DMO{\dipol}{D}
\DMO{\jkg}{J} 
\DMO{\vjkg}{VK} 
\DMO{\Tr}{Tr} 
\DMO{\Ind}{Ind} 
\DMO{\Zuo}{Mat} 
\DMO{\Pzuo}{PMat} 
\DMO{\St}{St} 
\DMO{\Ints}{Ints} 
\DMO{\Cov}{Cov} 
\DMO{\closse}{clo_{\sse}} 
\DMO{\clospe}{clo_{\supseteq}} 
\DMO{\edgemg}{ML} 
\DMO{\kneserg}{K} 
\DMO{\knesern}{\tau_0} 
\DMO{\nis}{nis} 
\DMO{\PBD}{PBD}
\nc{\BD}[1]{{#1}\text{-}\mr{BD}}
\DMO{\BIBD}{BIBD}
\DMO{\Steiner}{S}
\DMO{\SteinerTriple}{STS}
\DMO{\SteinerQuadruple}{SQS}
\DMO{\progeo}{PG} 
\DMO{\affgeo}{AG} 

\DMO{\astriv}{A_t} 
\DMO{\KochenSpecker}{KS}
\DMO{\KochenSpeckerErw}{KS'}
\DMO{\rankd}{rd}
\DMO{\mnconcomp}{mncc}
\DMO{\gpk}{\Box} 
\DMO{\gpw}{\times} 
\DMO{\gps}{\boxtimes} 
\DMO{\gjoin}{\boxdot} 
\DMO{\gjoinplus}{\boxplus} 
\DMO{\Ketten}{\mc{L}}
\DMO{\Antiketten}{\mc{A}}
\DMO{\comparable}{\Bumpeq}
\DMO{\incomparable}{\parallel}
\DMO{\pot}{\PP} 
\DMO{\pote}{\PP_f} 
\DMO{\potfv}{\overrightarrow{\PP}} 
\DMO{\potfvn}{\overrightarrow{\PP}^{\!*}} 
\DMO{\potfr}{\overleftarrow{\PP}} 
\DMO{\fak}{fac}
\DMO{\partitiont}{p}
\DMO{\teilt}{|}
\nc{\Prim}{\mc{PR}} 
\DMO{\ord}{ord}
\DMO{\ggt}{ggt}
\DMO{\kgv}{kgv}
\DMO{\opmod}{mod}
\DMO{\opdiv}{div}
\DMO{\eulphi}{\vp}
\newcommand{\Cls}{\mc{CLS}}

\nc{\Clsoo}{\Cls^{1,1}} 
\DMO{\munpuclash}{\mu{}NH}
\DMO{\hdef}{\delta_{\mr{h}}} 
\DMO{\rdef}{\delta_{\mr{r}}} 

\DMO{\nulli}{null} 
\DMO{\lit}{lit}
\DMO{\var}{var}
\DMO{\val}{val}
\DMO{\res}{\diamond} 
\DMO{\resop}{Res} 
\DMO{\mresop}{mRes} 
\DMO{\dpl}{DP} 

\DMO{\comp}{Comp} 
\DMO{\compex}{\comp_{ER}} 
\DMO{\compr}{\comp_R} 
\DMO{\comptr}{\comp_{tR}} 
\newcommand{\Us}{\mc{U}} 
\DMO{\comptru}{\comp_{tR(\Us)}} 
\DMO{\compru}{\comp_{R(\Us)}}

\DMO{\hardness}{hd}

\DMO{\pebf}{PF} 
\DMO{\rt}{rt} 
\DMO{\nds}{nds} 
\DMO{\lvs}{lvs} 
\DMO{\nlvs}{\#lvs} 
\DMO{\nnds}{\#nds} 
\DMO{\height}{ht}
\DMO{\depth}{d}
\DMO{\cls}{cls}
\DMO{\newcommandls}{\#cls}
\DMO{\ds}{ds}
\DMO{\dst}{ds_T}
\DMO{\dsg}{ds_G}
\DMO{\dpr}{dp}
\DMO{\dprt}{dp_T}
\DMO{\dprg}{dp_G}
\DMO{\ind}{in}
\DMO{\indg}{in_G}
\DMO{\outd}{out}
\DMO{\outdg}{out_G}
\DMO{\peb}{peb} 
\DMO{\taum}{\max \tau}
\DMO{\tauprob}{\tau^p} 
\DMO{\mtau}{\mf{T}} 
\DMO{\concatbt}{;} 
\DMO{\compobt}{\merge} 
\nc{\bth}[1]{\langle{#1}\rangle} 
\DMO{\pc}{pc}
\DMO{\aut}{Auk} 
\DMO{\laut}{LAuk} 
\DMO{\lautz}{LAuk_0} 
\DMO{\maut}{MAuk}
\newcommand{\A}{\mc{A}} 
\DMO{\nv}{N} 
\DMO{\na}{\nv_a} 
\DMO{\nA}{\nv_{\A}} 
\DMO{\nla}{\nv_{la}}
\DMO{\nbla}{\nv_{bla}}
\DMO{\nma}{\nv_{ma}}
\DMO{\npa}{\nv_{pa}}
\DMO{\baut}{BAuk} 
\DMO{\blaut}{BLAuk} 
\DMO{\blautz}{BLAuk_0} 
\DMO{\paut}{PAut} 
\DMO{\pautz}{PAut_0} 

\DMO{\resouz}{\overset{\Us, 0}{\vdash}}
\DMO{\resouo}{\overset{\Us, 1}{\vdash}}
\DMO{\resouk}{\overset{\Us,\, k}{\vdash}}
\DMO{\resou}{\,\overset{\Us}{\vdash}\,}
\DMO{\resour}{\,\overset{\Us_0}{\vdash}\,}
\DMO{\resourz}{\,\overset{\Us_0, 0}{\vdash}\,}
\DMO{\uresouk}{\resouk\hspace{-0.6em}\mbox{\raisebox{0.8ex}{\tiny u}}}
\DMO{\bresouk}{\resouk\hspace{-0.6em}\mbox{\raisebox{0.8ex}{\tiny b}}}
\DMO{\iresouk}{\resouk\hspace{-0.6em}\mbox{\raisebox{0.8ex}{\tiny i}}}
\DMO{\resok}{\overset{k}{\vdash}} 

\DMO{\wid}{wid} 
\DMO{\widl}{\hspace*{-1.5pt}wid}
\DMO{\widb}{\sideset{^{\mr{b}}}{}\widl}
\DMO{\widi}{\sideset{^{\mr{i}}}{}\widl}
\DMO{\cwid}{\mc{W}} 
\DMO{\cwidl}{\hspace*{-1pt}\mc{W}} 
\DMO{\cwidb}{\sideset{^{\mr{b}}}{}\cwidl}
\DMO{\cwidi}{\sideset{^{\mr{i}}}{}\cwidl}
\DMO{\modp}{mod_p} 
\DMO{\modt}{mod_t} 
\DMO{\moda}{\mf{S}} 
\DMO{\modf}{fal} 
\DMO{\mods}{mod} 
\DMO{\mus}{MU}
\DMO{\mss}{MS}
\DMO{\cmus}{CMU}
\DMO{\cmss}{CMS}
\DMO{\eqs}{EQ} 
\DMO{\neqs}{NEQ} 
\DMO{\scf}{CM} 
\DMO{\acf}{DCM} 
\DMO{\cmg}{cmg} 
\DMO{\cmdg}{cmdg} 
\DMO{\cg}{cg} 
\DMO{\gcg}{cgg} 
\DMO{\gcdg}{cgdg} 
\DMO{\rsg}{rg} 
\DMO{\srsg}{srg} 
\DMO{\vhg}{vhg} 
\DMO{\cvg}{cvg} 
\DMO{\cvmg}{cvmg} 
\DMO{\vig}{vig} 
\DMO{\vcg}{vcg} 
\DMO{\nscf}{bcp} 
\DMO{\nacf}{bcp_d} 
\DMO{\bcp}{bcp} 
\DMO{\tbcp}{tbcp} 
\DMO{\nsat}{\#sat}
\DMO{\nusat}{\#usat}
\DMO{\maxsat}{maxsat}
\DMO{\pmin}{\ul{rk}}
\DMO{\pmax}{rk}
\DMO{\pav}{\widetilde{rk}}
\DMO{\ldeg}{ldg}
\DMO{\minldeg}{\ul{ldg}}
\DMO{\maxldeg}{\ol{ldg}}
\DMO{\vdeg}{vdg}
\DMO{\minvdeg}{\ul{vdg}}
\DMO{\maxvdeg}{\ol{vdg}}
\DMO{\avvdeg}{\widetilde{vdg}}
\DMO{\cldeg}{cldg} 
\DMO{\mvardu}{\mu\!\vdeg}
\DMO{\Inj}{Inj}
\newcommand{\OKlibrary}{\texttt{OKlibrary}}

\newcommand{\OKinternet}{\url{http://www.ok-sat-library.org}}
\DMO{\Ex}{Ex} 
\nc{\ramz}[3]{\mr{ram}_{#1}^{#2}(#3)} 
\DMO{\ramzg}{ram} 
\nc{\waez}[2]{\mr{vdw}_{#1}(#2)} 
\DMO{\waezg}{vdw} 
\nc{\gtz}[2]{\mr{grt}_{#1}(#2)} 
\DMO{\gtzg}{grt} 
\DMO{\FvdW}{F_{W}} 
\DMO{\FRam}{F_{R}} 
\DMO{\arithp}{ap} 
\DMO{\arithpp}{ap_{pr}} 
\DMO{\crarithp}{cr_{ap}} 
\DMO{\crarithpp}{cr_{ap}^{pr}} 
\usepackage{theorem} 
\usepackage[hypertex]{hyperref}
\nc{\bm}{\boldmath}
\nc{\bmm}[1]{\mbox{\bm$\DST #1$}}
\nc{\mi}[1]{\bmm{\mathrm{(#1):}} \quad}
\usepackage[all]{xy}
\usepackage{enumerate}

\DeclareMathOperator{\gwaez}{N}
\DeclareMathOperator{\gcr}{cr}

\newcommand{\Fgt}[2]{\mr{F}^{\mr{GT}}_{#1}(#2)} 

\begin{document}

\pagestyle{headings}

\title{Exact Ramsey Theory: \\Green-Tao numbers and SAT}

\author{Oliver Kullmann}
\institute{
  Computer Science Department\\
  Swansea University\\
  \email{O.Kullmann@Swansea.ac.uk}\\
  \url{http://cs.swan.ac.uk/~csoliver}
}

\maketitle

\begin{abstract}
  We consider the links between Ramsey theory in the integers, based on van der Waerden's theorem, and (boolean, CNF) SAT solving. We aim at using the problems from exact Ramsey theory, concerned with computing Ramsey-type numbers, as a rich source of test problems, where especially methods for solving hard problems can be developed. We start our investigations here by reviewing the known \emph{van der Waerden numbers}, and we discuss directions in the parameter space where possibly the growth of van der Waerden numbers $\waez{m}{k_1,\dots,k_m}$ is only polynomial (this is important for obtaining feasible problem instances). We introduce \emph{transversal extensions} as a natural way of constructing mixed parameter tuples $(k_1, \dots, k_m)$ for van-der-Waerden-like numbers $\gwaez(k_1, \dots, k_m)$, and we show that the growth of the associated numbers is guaranteed to be linear. Based on Green-Tao's theorem (``the primes contain arbitrarily long arithmetic progressions'') we introduce the \emph{Green-Tao numbers} $\gtz{m}{k_1, \dots, k_m}$, which in a sense combine the strict structure of van der Waerden problems with the (pseudo-)randomness of the distribution of prime numbers. Using standard SAT solvers (look-ahead, conflict-driven, and local search) we determine the basic values. It turns out that already for this single form of Ramsey-type problems, when considering the best-performing solvers a wide variety of solver types is covered. For $m > 2$ the problems are non-boolean, and we introduce the \emph{generic translation scheme}, which offers an infinite variety of translations (``encodings'') and covers the known methods. In most cases the special instance called \emph{nested translation} proved to be far superior over its competitors (including the direct translation).
\end{abstract}

\section{Introduction}
\label{sec:intro}

The applicability of SAT solvers has made tremendous progress over the last 15 years; see the recent handbook \cite{2008HandbuchSAT}. We are concerned here with solving (concrete) combinatorial problems (see \cite{Zha09HBSAT} for an overview). Especially we are concerned with the computation of van-der-Waerden-like numbers, which is about colouring hypergraphs of arithmetic progressions.\footnote{This report is an extended version of \cite{Kullmann2010GreenTao}.}

An \emph{arithmetic progression} of size $k \in \NNZ$ in $\NN$ is a set $P \subset \NN$ of size $k$ such that after ordering (in the natural order), two neighbours always have the same distance. So the arithmetic progressions of size $k > 1$ are the sets of the form $P = \set{a + i \cdot d : i \in \tb{0}{k-1}}$ for $a, d \in \NN$. Van der Waerden's Theorem (\cite{vanderWaerden1927Baudet}) shows that whenever the set $\NN$ of natural numbers is partitioned into finitely many parts, some part must contain arithmetic progressions of arbitrary size. The finite version, which is equivalent to the above infinite version, says that for every progression size $k \in \NN$ and every number $m \in \NN$ of parts there exists some $n_0 \in \NN$ such that for $n \ge n_0$ every partitioning of $\tb 1n$ into $m$ parts has some part which contains an arithmetic progression of size $k$. The smallest such $n_0$ is denoted by \bmm{\waez{m}{k}}, and is called a \emph{vdW-number}. The subfield of Ramsey theory concerned with van der Waerden's theorem is for over 70 years now an active field of mathematics and combinatorics; for an elementary introduction see \cite{LandmanRobertson2003ArithmeticProgressions}.

We are concerned here with \emph{exact Ramsey theory}, that is, computing vdW-like numbers if possible, or otherwise producing (concrete) lower bounds. \cite{DransfieldLiuMarekTruszcynski2004VanderWaerden} introduced the application of SAT for computing vdW-numbers, showing that all known vdW-numbers (at that time) were rather easily computable with SAT solvers. With \cite{KourilPaulW26} yet SAT had its biggest success, computing the new (major) vdW-number $\waez{2}{6} = 1132$ (the problem of computing $\waez{2}{6}$ is mentioned in \cite{LandmanRobertson2003ArithmeticProgressions} as a difficult research problem). See \cite{Ahmed2009vdW,Ahmed2010vdW} for the current state-of-the-art. Regarding lower bounds, the best lower bounds currently one finds in \cite{HerwigHeuleLambalgenMaaren2005VanderWaerden}.\footnote{see \url{http://www.st.ewi.tudelft.nl/sat/waerden.php} for updates}

VdW-numbers for ``core'' parameter values (see Definition \ref{def:corep}) grow rapidly, and thus only few are known (see Section \ref{sec:thmGreenTao}). The first contribution of this article is the notion of a \emph{transversal extension} (see Definition \ref{def:corep}) of a parameter tuple, which allows to grow parameter tuples such that (only) linear growth of the associated vdW-numbers is guaranteed. The linear growth is proven in a general framework in Theorem \ref{thm:polygen}, and applied to vdW-numbers in Corollary \ref{cor:waezpoly}.

Next we introduce \emph{Green-Tao numbers} (``GT-numbers''; see Definition \ref{def:gtz}), which are defined as the vdW-numbers but using the first $n$ \emph{prime numbers} instead of the first $n$ natural numbers. The existence of these numbers is given by the celebrated Green-Tao Theorem (\cite{GreenTao2005Primes}). In Corollary \ref{cor:gtzpoly} we show that also for GT-numbers transversal extension numbers grow only linearly. In the remainder of the article we are concerned with computing GT-numbers.

For binary parameter tuples ($m=2$ above) the problems of computing vdW- or GT-numbers have a canonical translation to (boolean) SAT problems, while for $m > 2$ we still have a canonical translation into non-boolean SAT problems (as is the case in general for hypergraph colouring problems; see \cite{Kullmann2007ClausalFormZI}), but for using standard (boolean) SAT solvers the problem of a boolean translation arises. In Section \ref{sec:transnbb} we introduce the \emph{generic translation scheme}, with seven natural instances (amongst them the well-known direct and logarithmic translations). As it turns out, in nearly all cases for all solver types the \emph{weak nested translation} (introduced in \cite{Kullmann2007ClausalFormECCC2}) performed far best, with the only exception that for relatively large numbers of colours the logarithmic translation was better.

For this (initial) phase of investigations into GT-numbers we just used ``off-the-shelves'' SAT solvers, also aiming at some form of basic understanding why which type of solver is best on certain parameter ranges. For over one year on average 10 processors were running, with a lot of manual interaction and adjustment to find the right solvers and translations, and to set the parameters (most basic the number of vertices), establishing the basic Green-Tao numbers. All generators and the details of the computations are available in the open-source research platform \OKlibrary{} (see \cite{Kullmann2009OKlibrary}).\footnote{\OKinternet} See Section \ref{sec:compGreenTao} for the results of these computations. We conclude this article by a discussion of interesting research directions in Section \ref{sec:summary}.

\section{A few notions and notations}
\label{sec:notions}

We use $\NNZ = \ZZ_{\ge 0}$ and $\NN = \NNZ \sm \set{0}$. A \emph{finite hypergraph} $G$ is a pair $G = (V,E)$ where $V$ is a finite set and $E \sse \pot(V)$ (that is, $E$ is set of subsets of $V$); we use $V(G) := V$ and $E(G) := E$. An \emph{$m$-colouring} of a hypergraph $G$ is a map $f: V(G) \ra \tb 1m$ such that no hyperedge is monochromatic, that is, for every $H \in E(G)$ there are $v, w \in H$ with $f(v) \not= f(w)$. Regarding (boolean) clause-sets, complementation of boolean variables $v$ is denoted by $\ol{v}$, (boolean) clauses are finite and clash-free sets of (boolean) literals, and (boolean) clause-sets are finite sets of (boolean) clauses.

\section{The theorem of Green-Tao, and Green-Tao numbers}
\label{sec:thmGreenTao}

The numbers $\waez mk$ introduced in Section \ref{sec:intro} are ``diagonal vdW-numbers'', while we consider also the ``non-diagonal'' or \emph{mixed vdW-numbers}, which are defined as follows.

\begin{definition}\label{def:waez}
  A \textbf{parameter tuple} is an element of $\NN_{\ge 2}^m$ for some $m \in \NN$ which is monotonically non-decreasing (that is, sorted in non-decreasing order). For a parameter tuple $(k_1, \dots, k_m)$ the \textbf{vdW-number \bmm{\waez{m}{k_1, \dots, k_m}}} is the smallest $n_0 \in \NN$ such that for every $n \ge n_0$ and every $f: \tb 1n \ra \tb 1m$ there exists some ``colour'' (or ``part'') $i \in \tb 1m$ such that $f^{-1}(i)$ contains an arithmetic progression of size $k_i$.
\end{definition}
In a systematic study of parameter tuples and their operations, one likely should drop the sorting condition, and call our parameter tuples ``sorted'', however in this report we only consider sorted parameter tuples.

Obviously we have $\waez{m}{k_1, \dots, k_m} \le \waez{m}{\max(k_1, \dots, k_m)}$ (note that the right-hand side denotes a diagonal vdW-number), and thus also the mixed vdW-numbers exist (are always finite). The most up-to-date collection of precise (mixed) vdW-numbers one finds in \cite{Ahmed2009vdW,Ahmed2010vdW}.\footnote{The (updated) online-version is \url{http://users.encs.concordia.ca/~ta_ahmed/vdw.html}.} For the sake of completeness we state the numbers here, but for references we refer to \cite{Ahmed2009vdW,Ahmed2010vdW}. We introduce the following organisation of the parameter space.
\begin{definition}\label{def:corep}
  A parameter tuple is \textbf{trivial} if all entries are equal to $2$, otherwise it is \textbf{non-trivial}. A \textbf{simple} parameter tuple has length $1$, otherwise it is \textbf{non-simple}. A parameter tuple is a \textbf{core tuple} if it is non-simple and if all entries are greater than or equal to $3$. A parameter tuple $t$ is a \textbf{(transversal) extension} of a parameter tuple $t'$ if $t$ can be obtained from $t'$ by adding entries equal to $2$ to the front of $t'$. A transversal extension of a simple parameter tuple is called an \textbf{extended simple tuple} or a \textbf{transversal tuple}, while an transversal extension of a core tuple is called an \textbf{extended core tuple}. Finally a parameter tuple is \textbf{diagonal}, if it is constant (all entries are equal), while otherwise it is \textbf{non-diagonal} or \textbf{mixed}.

  Accordingly we speak of (and are interested in) \textbf{trivial vdW-numbers}, \textbf{simple vdW-numbers}, \textbf{core vdW-numbers}, \textbf{transversal vdW-numbers}, \textbf{extended core vdW-num\-bers}, and \textbf{diagonal vdW-numbers}.
\end{definition}
The trivial vdW-numbers are $\waez{m}{2} = m+1$, while the simple vdW-numbers are given by $\waez{1}{k} = k$. The known core vdW-numbers are as follows.
\begin{enumerate}
\item $24$ binary core vdW-numbers $\waez{2}{a,b}$ are known:
  \begin{center}
    \hspace*{-1em}
    \begin{tabular}[c]{c||*{16}{c@{ $\;$}}}
      $\begin{array}[r]{c@{\hspace{0.5em}}r}
        & b\\[-1.9ex]
        a
      \end{array}\hspace{-0.1em}$ & $3$ & $4$ & $5$ & $6$ & $7$ & $8$ & $9$ & $10$ & $11$ & $12$ & $13$ & $14$ & $15$ & $16$ & $17$ & $18$\\
      \hline\hline
      $3$ & $9$ & $18$ & $22$ & $32$ & $46$ & $58$ & $77$ & $97$ & $114$ & $135$ & $160$ & $186$ & $218$ & $238$ & $279$ & $312$\\
      $4$ & - & 35 & 55 & 73 & 109 & 146\\
      $5$ & - & - & $178$ & $206$\\
      $6$ & - & - & - & $1132$
    \end{tabular}
  \end{center}
\item $4$ core vdW-numbers $\waez{3}{a,b,c}$ and one core vdW-number $\waez{4}{a,b,c,d}$ are known:
  \begin{center}
    \begin{tabular}[c]{c||*{3}{c}}
      $\begin{array}[r]{c@{\hspace{0.5em}}r}
        & c\\[-1.9ex]
        a,b
      \end{array}\hspace{-0.1em}$ & $3$ & $4$ & $5$\\
      \hline\hline
      $3,3$ & $27$ & $51$ & $80$\\
      $3,4$ & - & 89
    \end{tabular} , \quad
    \begin{tabular}[c]{c||*{1}{c}}
      $\begin{array}[r]{c@{\hspace{0.5em}}r}
        & d\\[-1.9ex]
        a,b,c
      \end{array}\hspace{-0.1em}$ & $3$\\
      \hline\hline
      $3,3,3$ & $76$
    \end{tabular} .
  \end{center}
\end{enumerate}

A basic quest for this article is in what ``directions'' can one move through the parameter space while experiencing only polynomial growth? Regarding vdW-numbers we consider the conjecture (perhaps better called ``question''), that extended core parameter tuples grow only polynomially in the extension length, where for parameter tuples $a, b$ by $a;b$ we denote their concatenation:
\begin{conjecture}\label{conj:onecomp}
  For every (fixed) parameter tuple $t$ of length $m$ the map $k \in \NN \mapsto \waez{m+1}{t;(k)}$ is polynomially bounded in $k$ (depending on $t$).
\end{conjecture}
In other words, for the above tables and all similarly constructed tables growth in every row is polynomially bounded. The evidence for Conjecture \ref{conj:onecomp} is as follows.
\begin{enumerate}
\item The case $t=(3)$, and more precisely $\waez{2}{3,k} \le k^2$, has been suggested in \cite{BrownLandmanRobertson2008WaerdenBounds}. The numbers $\waez{2}{3,k}$ are known for $1 \le k \le 18$ (see above). Additionally, our experiments yield the following conjectured values (where using ``$\ge x$'' means that we believe that actually equality holds, while the lower bound ``$> x-1$'' has been shown), further supporting the conjectured upper bound:\footnote{All lower bounds are obtained by local-search algorithms from the Ubcsat-suite (see \cite{TompkinsHoos2004Ubcsat}), and all data is available through the \OKlibrary.} $\waez{2}{3,19} \ge 349$, $\waez{2}{3,20} \ge 389$, $\waez{2}{3,21} \ge 416$.
\item Considering $t=(4)$, the numbers $\waez{2}{4,k}$ are known for $1 \le k \le 8$ (see above), while the bound $\waez{2}{4,9} > 254$ is in \cite{Ahmed2009vdW}; we can improve this to $\waez{2}{4,9} \ge 309$, and furthermore $\waez{2}{4,10} > 328$. So going from $k=8$ to $k=9$ we see a rather big jump, however possibly from $k=9$ to $k=10$ only a small change might take place.
\item For general $t = (k_0)$ with $k_0 \ge 3$, in \cite{BrownLandmanRobertson2008WaerdenBounds} the lower bound $\waez{2}{k_0,k} \ge k^{k_0 - 1 - \log(\log(k))}$ for sufficiently large $k$ has been shown. It seems consistent with current knowledge that we could have $\waez{2}{k_0,k} \le k^{k_0-1}$ for all $k, k_0 \ge 1$.
\end{enumerate}
A contribution of this article is the systematic consideration of transversal extensions as defined in Definition \ref{def:corep}. The known $33 + 10 + 1 + 6 = 50$ extended core vdW-numbers are as follows (again, for references see \cite{Ahmed2009vdW}); transversal vdW-numbers are presented in Section \ref{sec:transvdwnumbers}; see Subsection \ref{sec:remarkstrans} for general remarks.
\begin{enumerate}
\item Extending $(3,k)$ by $m$ $2$'s, i.e., the numbers $\waez{m+2}{2, \dots, 2, 3,k}$:
  \begin{center}
    \begin{tabular}[c]{c||*{11}{c@{ $\;$}}}
      $\begin{array}[r]{c@{\hspace{0.5em}}r}
        & k\\[-1.9ex]
        m
      \end{array}\hspace{-0.1em}$ & $3$ & $4$ & $5$ & $6$ & $7$ & $8$ & $9$ & $10$ & $11$ & $12$ & $13$\\
      \hline\hline
      $0$ & 9 & 18 & 22 & 32 & 46 & 58 & 77 & 97 & 114 & 135 & 160\\
      \hline
      $1$ & 14 & 21 & 32 & 40 & 55 & 72 & 90 & 108 & 129 & 150 & 171\\
      $2$ & 17 & 25 & 43 & 48 & 65 & 83 & 99 & 119\\
      $3$ & 20 & 29 & 44 & 56 & 72 & 88\\
      $4$ & 21 & 33 & 50 & 60\\
      $5$ & 24 & 36\\
      $6$ & 25\\
      $7$ & 28
    \end{tabular}
  \end{center}
\item Extending $(4,k)$ resp.\ $(5,k)$ by $m$ $2$'s, i.e., numbers $\waez{m+2}{2, \dots, 2, 4,k}$ resp.\ $\waez{m+2}{2, \dots, 2, 5,k}$:
  \begin{center}
    \begin{tabular}[c]{c||*{5}{c@{ $\;$}}}
      $\begin{array}[r]{c@{\hspace{0.5em}}r}
        & k\\[-1.9ex]
        m
      \end{array}\hspace{-0.1em}$ & $4$ & $5$ & $6$ & $7$ & $8$ \\
      \hline\hline
      $0$ & 35 & 55 & 73 & 109 & 146\\
      \hline
      $1$ & 40 & 71 & 83 & 119\\
      $2$ & 53 & 75 & 93\\
      $3$ & 54 & 79\\
      $4$ & 56
    \end{tabular} , \quad
    \begin{tabular}[c]{c||*{1}{c@{ $\;$}}}
      $\begin{array}[r]{c@{\hspace{0.5em}}r}
        & k\\[-1.9ex]
        m
      \end{array}\hspace{-0.1em}$ & $5$ \\
      \hline\hline
      $0$ & 178\\
      \hline
      $1$ & 180
    \end{tabular} .
  \end{center}
\item Extending $(3,3,k)$ by $m$ $2$'s, i.e., numbers $\waez{m+3}{2, \dots, 2, 3,3,k}$:
  \begin{center}
    \begin{tabular}[c]{c||*{4}{c@{ $\;$}}}
      $\begin{array}[r]{c@{\hspace{0.5em}}r}
        & k\\[-1.9ex]
        m
      \end{array}\hspace{-0.1em}$ & $3$ & $4$ & $5$ \\
      \hline\hline
      $0$ & 27 & 51 & 80\\
      \hline
      $1$ & 40 & 60 & 86\\
      $2$ & 41 & 63\\
      $3$ & 42\\
    \end{tabular}
  \end{center}
\end{enumerate}
Note that by Conjecture \ref{conj:onecomp} we would have in every row only polynomial growth. Now in Corollary \ref{cor:waezpoly} we will prove that in every column we have \emph{linear growth}, where actually the factor can be made as close to $1$ as one wishes, when only $m$ is big enough.

\subsection{A general perspective on Ramsey theory}
\label{sec:Ramseygeneral}

We consider a sequence $(G_n)_{n \in \NN}$ of finite hypergraphs, where we assume that we have $V(G_n) \sse V(G_{n+1})$ and $E(G_n) \sse E(G_{n+1})$ for all $n$. Furthermore we assume $V(G_1) \not= \es$ and $\fa\, n \in \NN : \es \notin E(G_n)$ for simplicity. Such a sequence of hypergraphs we call \emph{nontrivial monotonic}. We consider the following questions:
\begin{enumerate}[(i)]
\item Does there exist some $n \in \NN$ with $E(G_n) \not= \es$ ?
\item Does for every $m \in \NN$ exists some $n_0(m) \in \NN$ such that for all $n \ge n_0$ the hypergraph $G_n$ is not $m$-colourable? In this case we say that $(G_n)_{n \in \NN}$ has the \emph{Ramsey property}.
\item Does $\lim_{n \ra \infty} \frac{\alpha(G_n)}{\abs{V(_n)}} = 0$ hold, where $\alpha(G)$ for a hypergraph $G$ is the \emph{independence number} of $G$, the maximum size of an independent vertex set (not containing any hyperedge)? In this case we say that $(G_n)_{n \in \NN}$ has the \emph{Szemer\'edi property}.
\end{enumerate}
Clearly (ii) implies (i), while in turn (iii) implies (ii), since colouring a hypergraph $G$ with $m$ colours just means to partition $V(G)$ into at most $m$ independent subsets. Considering the original vdW-problem, we have $G_n = \arithp(k,n)$ for some fixed $k \in \NN$, where $V(\arithp(k,n)) = \tb 1n$, while $E(\arithp(k,n))$ is the set of arithmetic progressions of size $k$ in $\tb 1n$. Property (i) trivially holds, while property (ii) is van der Waerden's theorem. And property (iii) has been conjectured by Erd\"os and Tur\'an in 1936 (\cite{ErdoesTuran1936vdW}), and was finally proved by Szemer\'edi in his landmark paper \cite{Szemeredi1975AP} (for arbitrary $k$, one of the deepest results in combinatorics; for $k=3$ it was proven in \cite{Roth1953vdW}, for $k=4$ in \cite{Szemeredi1969AP}).

Consider a nontrivial monotonic sequence $G = (G_n)_{n \in \NN}$ of hypergraphs. The following definition generalises diagonal vdW-numbers, and introduces a form of ``convergence rate'' capturing the Szemer\'edi property.
\begin{definition}\label{def:genvdwN}
  For $m \in \NN$ let $\bmm{\gwaez_m(G)} \in \NN \cup \set{+\infty}$ be the infimum of $n \in \NN$ such that $G_n$ is not $m$-colourable. And for $q \in \RR_{>0}$ let $\bmm{\gcr(G,q)} \in \NN \cup \set{+\infty}$ be the infimum of $n \in \NN$ such that for all $n' \ge n$ holds $\frac{\alpha(G_{n'})}{\abs{V(G_{n'})}} < q$.
\end{definition}
Thus $G$ has the Ramsey property iff for all $m \in \NN$ we have $\gwaez_m(G) < + \infty$, while $G$ has the Szemer\'edi property iff for all $q \in \ioa 01$ we have $\gcr(G,q) < +\infty$. Considering the sequence $(\arithp(k,n))_{n \in \NN}$ of vdW-hypergraphs of arithmetic progressions of size $k$ we have $\gwaez_m(\arithp(k,-)) = \waez mk$. The following simple fact makes the above remark, that (iii) implies (ii), more precise.
\begin{lemma}\label{lem:genupb}
  For all $m \in \NN$ we have $\gwaez_m(G) \le \gcr(G,\frac 1m)$.
\end{lemma}
We say that hypergraph sequences $G^1, \dots, G^m$ are \emph{compatible} if for all $n$ we have $V(G^1_n) = \dots = V(G^m_n)$. Generalising the notion of ``diagonal vdW-like numbers'' in Definition \ref{def:genvdwN} and the notion of mixed vdW-numbers in Definition \ref{def:waez}:
\begin{definition}\label{def:ggenvdewN}
  Consider $m \in \NN$ and compatible nontrivial monotonic hypergraph sequences $G^1, \dots, G^m$. Then $\bmm{\gwaez_m(G^1, \dots, G^m)} \in \NN \cup \set{+\infty}$ is defined as the infimum of $n \in \NN$ such that for every $m$-colouring of $V(G^1_n)$ there exists some $i \in \tb 1m$ such that some hyperedge of $G^i_n$ is monochromatically $i$-coloured.
\end{definition}
Obviously we have $\gwaez_m(G) = \gwaez_m(G, \dots, G)$. Call $(G^1, \dots, G^m)$ \emph{horizontally monotonic} if for all $n \in \NN$ and all $1 \le i \le j \le m$ every independent subset of $G^i_n$ is also independent in $G^j_n$. In this case then $\gwaez_m(G^1, \dots, G^m) \le \gwaez_m(G_m)$ holds. This captures the typical application of ``mixed numbers'' from Ramsey theory. Generalising the notion of ``convergence rate'' in Definition \ref{def:genvdwN}:
\begin{definition}\label{def:ggcr}
  Consider $m \in \NN$ and compatible nontrivial monotonic hypergraph sequences $G^1, \dots, G^m$. For $q \in \RR_{>0}$ let $\bmm{\gcr((G_1,\dots,G_m),q)} \in \NN \cup \set{+\infty}$ be the infimum of $n \in \NN$ such that for all $n' \ge n$ and for all $m$-tuples $(S_1, \dots, S_m)$ of (pairwise) disjoint independent subsets $S_i$ of $G^i_{n'}$ we have $\frac{\abs{S_1} + \dots + \abs{S_m}}{n'} < q$.
\end{definition}
A few basic remarks:
\begin{enumerate}
\item\label{item:ggcr1} $\gcr(G,q) = \gcr((G),q)$.
\item\label{item:ggcr2} $\gcr((G^1, \dots, G^m), q) \le \gcr(G_m, \frac qm)$ if $(G^1, \dots, G^m)$ is horizontally monotonic.
\item\label{item:ggcr3} By definition we have for arbitrary compatible nontrivial monotonic hypergraph sequences that $\gwaez_m(G^1, \dots, G^m) = \gcr((G^1, \dots, G^m), 1)$. By Remark \ref{item:ggcr2}) this generalises Lemma \ref{lem:genupb}.
\end{enumerate}
For complete hypergraphs we can easily establish the Szemer\'edi property:
\begin{lemma}\label{lem:gcrvoll}
  For $n, k \in \NN$ let $V^k_n$ be the hypergraph with vertex set $\tb 1n$ and hyperedge set $\binom{\tb 1n}{k}$. Now for natural numbers $k_1, \dots, k_m$ and $q > 0$ we have that $\gcr((V^{k_1}_{-}, \dots, V^{k_m}_{-}), q)$ is the smallest $n > \frac{(\sum_{i=1}^m k_i)-m}{q}$. Especially we have $\gcr((V^2, \dots, V^2), q) = \gcr(V^{m+1},q)$, which is the smallest $n > \frac mq$.
\end{lemma}
By definition we get the following generalisation of Remark \ref{item:ggcr2} to Definition \ref{def:ggcr}:
\begin{lemma}\label{lem:gcradd}
  For $m \in \NN$ consider compatible nontrivial monotonic hypergraph sequences $G^1, \dots, G^m$, and consider $1 \le t < m$. Then for $p, q \in \RR_{>0}$ we have
  \begin{displaymath}
    \gcr((G^1, \dots, G^m), p+q) \le \max \big( \gcr((G^1, \dots, G^t), p), \, \gcr((G^{t+1}, \dots, G^m), q) \big).
  \end{displaymath}
\end{lemma}
Lemma \ref{lem:gcrvoll} and Lemma \ref{lem:gcradd} (splitting $1 = \frac 1s + (1 - \frac 1s)$) together yield the basic theoretical observation of this paper:
\begin{theorem}\label{thm:polygen}
  Consider $l \in \NN$ and compatible nontrivial monotonic hypergraph sequences $G^1, \dots, G^l$. For $n \in \NN$ let $V_n := V(G^1_n)$, and for $k \in \NN$ let $Q_n^k := (V_n, \binom{V_n}k)$, and thus $Q^k = (Q^k_n)_{n \in \NN}$ is a nontrivial monotonic hypergraph sequence. For $x \in \RR$ let $M(x) \in \NN \cup \set{+\infty}$ be the infimum of $n \in \NN$ such that we have $\abs{V_n} > x$. Now for every $s \in \RR$ with $s > 1$ and for every $m \in \NNZ$ we have
  \begin{gather*}
    \gwaez_{l+m}(G^1, \dots, G^l, Q^2, \dots, Q^2) = \gwaez_{l+1}(G^1, \dots, G^l, Q^{m+1}) \le\\
    \max \Big (M(s \cdot m), \, \gcr \big((G^1, \dots, G^l), 1 - \frac 1s \big) \Big ).
  \end{gather*}
\end{theorem}
So the growth-rate of $m \mapsto \gwaez_{l+m}(G^1, \dots, G^l, Q^2, \dots, Q^2)$ is linear for $m$ large enough, where the factor can be made arbitrarily close to $1$. Applied to vdW-numbers, using Szemer\'edi's theorem, we get the following application (the proof-idea here originated from Jan-Christoph Schlage-Puchta). As a special case of Definition \ref{def:ggcr} we use $\crarithp(t, q)$ for parameter tuples $t$, using the hypergraph sequences belonging to the progression sizes in $t$.
\begin{corollary}\label{cor:waezpoly}
  For a parameter tuple $t$ of length $l \in \NN$, for $m \in \NNZ$ and for $s \in \RR_{>1}$ we have $\waez{m+l}{(2, \dots, 2); t} \le \max(s \cdot m + 1, \crarithp(t, 1 - \frac 1s))$.
\end{corollary}

Giving up on the factor, but now without unknown minimal value for $m$, we have the following variation on Corollary \ref{cor:waezpoly}:
\begin{lemma}\label{lem:waezpolys}
  For a parameter tuple $t$ of length $l \in \NN$ and for $m \in \NNZ$ we have $\waez{m+l}{(2, \dots, 2); t} \le (m+1) \cdot \waez{l}{t}$.
\end{lemma}
\begin{proof}
  Let $n := (m+1) \cdot \waez{l}{t}$. Now for any $S \sse \tb 1n$ with $\abs{S} \le m$ the set $\tb 1n \sm S$ contains at least one interval $\tb ij$ for $1 \le i \le j \le n$ with $j - i + 1 = \waez lt$. Using the invariance of linear progressions under translation, we obtain the desired inequality. \qed
\end{proof}

\subsection{Arithmetic progressions in the prime numbers}
\label{sec:appn}

We turn to a major strengthening of Szemer\'edi's theorem. Now the hypergraph sequence is given as $G_n = \arithpp(k,n)$ for fixed $k \in \NN$, where the vertex set of $\arithpp(k,n)$ is the set of the first $n$ prime numbers, while the hyperedges are the arithmetic progressions of size $k$ (within the first $n$ prime numbers). As before, every set of prime numbers having at most two elements is an arithmetic progression, but now the first arithmetic progression of size $3$ is $\set{3,5,7}$, and the first arithmetic progression of size $4$ is $\set{5,11,17,23}$. Until 2004 even condition (i) was unknown, that is, whether the primes contain arbitrarily long arithmetic progressions, and only with \cite{GreenTao2005Primes} not only condition (i) was proven, but even condition (iii) (the underlying preprint was a major contribution towards the Fields medal for Terence Tao in 2006). Actually, until today no other proof of property (i) is known than through property (iii)! In analogy to Definition \ref{def:waez}, and as a special case of Definition \ref{def:ggenvdewN}, we define \emph{Green-Tao numbers} (``GT-numbers'').
\begin{definition}\label{def:gtz}
  For a parameter tuple $(k_1, \dots, k_m)$ let the \textbf{Green-Tao number} \bmm{\gtz{m}{k_1, \dots, k_m}} be defined as the smallest $n_0 \in \NN$ such that for every $n \ge n_0$ and every $f: \set{p_1, \dots, p_n} \ra \tb 1m$, where $p_1, \dots, p_n$ are the first $n$ prime numbers, there exists some $i \in \tb 1m$ such that $f^{-1}(i)$ contains an arithmetic progression of size $k_i$.

  According to Definition \ref{def:corep} we speak of \textbf{trivial GT-numbers}, \textbf{simple GT-numbers}, \textbf{core GT-numbers}, \textbf{transversal GT-numbers}, \textbf{extended core GT-num\-bers}, and \textbf{diagonal GT-numbers}.
\end{definition}
Theorem \ref{thm:polygen} applied to Green-Tao numbers, using Green-Tao's theorem (\cite{GreenTao2005Primes}), yields that extended GT-numbers grow linearly. For the explicit statement, as for Corollary \ref{cor:waezpoly} and as a special case of Definition \ref{def:ggcr}, we use $\crarithpp(t, q)$ for parameter tuples $t$, using the hypergraph sequences in the primes belonging to the progression sizes in $t$.
\begin{corollary}\label{cor:gtzpoly}
   For a parameter tuple $t$ of length $l \in \NN$, for $m \in \NNZ$ and for $s \in \RR_{>1}$ we have $\gtz{m+l}{(2, \dots, 2); t} \le \max(s \cdot m + 1, \crarithpp(t, 1 - \frac 1s))$.
\end{corollary}

\subsection{Remarks on transversal numbers and transversal extensions}
\label{sec:remarkstrans}

Given a nontrivial monotonic hypergraph sequence $G$, the (computational) determination of the simplest transversal extension numbers $\gwaez_{1+m}(G, Q^2, \dots, Q^2) = \gwaez_2(G, Q^{m+1})$ (recall Theorem \ref{thm:polygen}), with the special cases $\waez{m+1}{2,\dots,2,k}$ and $\gtz{m+1}{2,\dots,2,k}$, is relatively(!) easy, since essentially we have to compute the \emph{transversal numbers} $\tau(G_n)$ of the hypergraphs $G_n$ (though still an NP-complete task in general), that is the minimum size of a set of vertices having non-empty intersection with every hyperedge. This is also the motivation for the notion of ``transversal $N$-number'' and ``transversal extension'': $\gwaez_2(G, Q^{m+1})$ is the smallest $n$ such that $\tau(G_n) > m$. The complements of independent sets in a hypergraph $G$ are exactly the transversals of $G$, and thus $\tau(G) + \alpha(G) = \abs{V(G)}$ holds. So determination of the transversal numbers for the hypergraph sequence $G$ determines the convergence rate w.r.t.\ the Szemer\'edi property, and is therefore of strong interest (recall Lemma \ref{lem:genupb}).

Considering the computation of $\tau(G_n)$ for the vdW- and the GT-sequence of hypergraphs, going from $G_n$ to $G_{n+1}$ only one vertex is added, and thus we have a relatively slow growth of complexity compared to transversal extensions of core tuples, without a clear boundary of what becomes ``infeasible''. These problems also require some special treatment (using cardinality constraints or special hypergraph transversal algorithms). So we put the results on transversal vdW- or GT-numbers only into the appendix (see Section \ref{sec:Transversalnumbers}), where we used the most direct method for computing transversal numbers of hypergraphs via SAT (see the introduction to Section \ref{sec:Transversalnumbers}).
\begin{itemize}
\item Special methods are applicable regarding the transversal numbers of vdW-hypergraphs, which exploit the translation invariance of arithmetic progressions; see \cite{Wagstaff1967ArithProg,Wagstaff1972ArithProg} for the basic ideas, and see the case $k=3$ in Subsection \ref{sec:transvdwnumbers} for data derived by such special methods. \cite{LandmanRobertsonCulver2005vanderWaerden} even found for ``small'' $m$ a precise formula for $\waez{m+1}{2,\dots,2,k}$.
\item Regarding $\gtz{m+1}{2,\dots,2,k}$, such a (simple) formula likely does not exist, and also we loose translation invariance of the arithmetic progressions (since they must lie in the primes), so computing the minimum size of hypergraph transversals via SAT solving seems a good option, but still should (and can) exploit special properties (not investigated in this paper).
\item It seems that combining these special methods with SAT solving should yield the best results.
\end{itemize}
Finally we mention that also transversal extensions of core tuples can be translated into SAT problems by combining the general translation methods of the following section with cardinality constraints (which take care of the initial tuples of $2$'s). In this article we concentrated on the foundations and on the study of the various general translation schemes, so also the investigations of this special treatment had to be postponed (that is, the GT-numbers for transversal extensions reported in Section \ref{sec:compGreenTao} have been obtained by just applying the general translations of non-boolean problems into boolean problems).

\section{The generic translation scheme from non-boolean clause-sets to boolean clause-sets}
\label{sec:transnbb}

GT-problems of the form ``$\gtz{2}{k_1,k_2} > n$ ?'' have a natural formulation as (boolean) SAT problems by just excluding the arithmetic progressions of sizes $k_1$ and $k_2$, e.g.\ the problem ``$\gtz{2}{2,3} > 4$ ?'' yields the (satisfiable) clause-set $\setb{\set{2,3},\set{2,5},\set{2,7},\set{3,5},\set{3,7},\set{5,7},\set{-3,-5,-7}}$ over the variable-set $\set{2,3,5,7}$ (thus the answer is ``yes''). A natural translation for arbitrary $m$ is given when using \emph{generalised clause-sets} as systematically studied in \cite{Kullmann2007ClausalFormECCC2,Kullmann2007ClausalFormZI,Kullmann2007ClausalFormZII}, which allow variables $v$ with finite domains $D_v$ and literals of the form ``$v \not= \ve$'' for values $\ve \in D_v$. The problem of colouring a hypergraph $G$ with $m$ colours is naturally translated into a SAT problem for generalised clause-sets via using $m$ clauses for every hyperedge $H \in E(G)$, namely for every value $\ve \in \tb 1m$ the clause $\set{v \not= \ve : v \in H}$, stating that not all vertices in $H$ can have value $\ve$ (note that the vertices of $G$ are used as variables with (uniform) domain $\tb 1m$). Accordingly we arrive at the natural generalisation $\Fgt{k_1, \dots, k_m}{n}$ of the boolean formulation, using as variables the first $n$ prime numbers, each with domain $\tb 1m$, where the clauses are obtained from the hyperedges of $\arithpp(k_i,n)$ for $i \in \tb 1m$ by using literals ``$v \not= i$''.

As a running example consider $m=3$, $k_1=k_2=k_3=3$ and $n=5$. We remark that we have $\gtz{3}{3} = 137$, as can be seen in Section \ref{sec:compGreenTao}. Only one hypergraph needs to be considered here (since all $k_i$-values coincide), namely $\arithpp(3,5) = (\set{2,3,5,7,11}, \set{\set{3,5,7},\set{3,7,11}})$. Now the (non-boolean) clause-set $\Fgt{3,3,3}{5}$ uses the five (formal\footnote{note that variable $2$ does not occur here; it occurs only for $k_i = 2$, and one could ignore it in general, however then we always had to use the offset $1$ when comparing with prime number tables}) variables $2,3,5,7,11$, each with domain $\set{1,2,3}$, while we have $3 \cdot 2 = 6$ clauses (each of length $3$), namely the clauses $\set{(3,i),(5,i),(7,i)}, \set{(3,i),(7,i),(11,i)}$ for $i \in \set{1,2,3}$.

In \cite{Kullmann2007ClausalFormECCC2} the \emph{nested translation} from generalised clause-sets to boolean clause-sets was introduced, while the generalisation to the \emph{generic translation scheme} is outlined in \cite{Kullmann2007ClausalFormZII}. Given a generalised clause-set $F$, for every variable an (arbitrary) unsatisfiable boolean clause-set $T(v)$ is chosen, such that for different variables these clause-sets are variable-disjoint. Furthermore for every value $\ve \in D_v$ a necessary clause $\gamma_v(\ve) \in T(v)$ is chosen (that is, $T(v) \sm \set{\gamma_v(\ve)}$ is satisfiable), such that to different values different clauses are assigned. Now the translation $T_{\gamma}(F)$ of $F$ under $T$ and $\gamma$ replaces for every clause $C \in F$ the (non-boolean) literals $v \not= \ve$ by the (boolean) literals in clause $\gamma_v(\ve)$, and adds for every variable $v \in \var(F)$ the clauses of the (boolean) clause-set $T(v) \sm \set{\gamma_v(\ve) : \ve \in D_v}$. The clauses $\gamma_v(\ve)$ are called the \emph{main clauses} of $T(v)$, while the other clauses of $T(v)$ constitute the \emph{remainder}.

\begin{lemma}\label{lem:corrTrans}
  $T_{\gamma}(F)$ is satisfiability-equivalent to $F$.
\end{lemma}
\begin{proof}
  If $\vp$ is a satisfying assignment for $F$, then for every variable $v \in \var(\vp)$ choose a satisfying assignment $\psi_v$ of $T(v) \sm \set{\gamma_v(\vp(v))}$, and the union of these (compatible) assignments $\psi_v$ yields a satisfying assignment for $T_{\gamma}(F)$ (here it is used that for $\ve \in D_v \sm \set{\vp(v)}$ we have $\gamma_v(\ve) \not= \gamma_v(\vp(v))$). If on the other hand $\psi$ is a satisfying (total) assignment for $T_{\gamma}(F)$, then for every clause-set $T(v)$ there exists some $\ve_v \in D_v$ such that the clause $\gamma_v(\ve)$ is falsified by $\psi$; now the assignment $v \mapsto \ve_v$ satisfies $F$. \qed
\end{proof}
The seven instances of the generic scheme used in this paper, where the domain of variable $v$ is $\tb 1m$, and where the boolean variables are $v_i$ for appropriate indices $i$, are as follows:
\begin{enumerate}
\item $T(v) = D_m := \setb{\set{v_1}, \dots, \set{v_m}, \set{\ol{v_1}, \dots, \ol{v_m}}}$ with $m$ variables is used for the \emph{weak direct translation}, where $\gamma_v(i) := \set{v_i}$. $D_m$ is a marginal minimally unsatisfiable clause-set\footnote{See \cite{Kullmann2007HandbuchMU} for an overview on minimally unsatisfiable clause-sets.} with deficiency $1$ (that is, with $m+1$ clauses). The \emph{strong direct translation} uses $T(v) = D'_m := D_m \cup \set{ \set{v_i, v_j} : 1 \le i < j \le m }$ and the same $\gamma_v$.\footnote{In \cite{Pre09HBSAT} the ``strong direct translation'' is called ``direct encoding'', starting from arbitrary CSP-problems (while we start from generalised clause-sets). We prefer to distinguish between ``encodings'', which are about variables and the mapping of assignments, and ``translations'', which concern the whole process, and which can use quite different but semantically equivalent clause-sets for example. For the direct translation it seems that always the strong form is better, but this is not the case for other translations, and so we explicitely distinguish between ``weak'' and ``strong''.}
\item The \emph{weak reduced translation} uses $m-1$ variables with $T(v) = D_{m-1}$ and an arbitrary bijection $\gamma_v$ (note that $D_{m-1}$ has $m$ clauses), while the \emph{strong reduced translation} uses the same $\gamma_v$ and $T(v) = D'_{m-1}$. Different from the direct translations, here $\gamma_v$ plays a role now, namely the question is to which value one associates the long clause $\set{\ol{v_1}, \dots, \ol{v_{m-1}}}$, and so we have $m$ (essentially) different choices.

  Note that clause-set $D_{m-1}$ can be obtained from $D_m$ by DP-reduction for variable $v_m$ (replacing all clauses containing variable $v_m$ by their resolvents on $v_m$), and accordingly from a clause-set translated by the (weak/strong) direct translation we obtain the clause-set translated by the (weak/strong) reduced translation by performing DP-reduction on all such variables $v_m$ (using that the remainder-clauses are just used as they are, without additional literals in them).
\item The \emph{weak nested translation} uses $m-1$ variables and $T(v) = H_{m-1}$, where
  \begin{displaymath}
    H_m := \setb{ \set{v_1}, \set{\ol{v_1}, v_2}, \dots, \set{\ol{v_1}, \dots, \ol{v_{m-1}}, v_m}, \set{\ol{v_1}, \dots, \ol{v_m}}},
  \end{displaymath}
  using some arbitrary bijection $\gamma_v$ (note that $H_m$ has deficiency $1$, and thus $H_{m-1}$ has $m$ clauses). $H_m$ is up to isomorphism the unique saturated minimally unsatisfiable Horn clause-set with $m$ variables, and in fact is a saturation of the minimally unsatisfiable clause-set $D_m$ (see \cite{Kullmann2007HandbuchMU}). The \emph{strong nested translation} uses the same $\gamma_v$, and, similar to the strong direct translation, $T(v) = H'_{m-1} := H_{m-1} \cup \set{ \set{v_i, v_j} : 1 \le i < j \le m-1 }$. For both forms now we have $m!/2$ (essentially) different choices for $\gamma_v$ (note that only the two clauses of length $m$ in $H_m$ can be mapped to each other by an isomorphism of $H_m$). The motivation for the introduction of the weak nested translation in \cite{Kullmann2007ClausalFormECCC2,Kullmann2007ClausalFormZII} was that first the number of clauses is not changed by the translation, that is, $T(v)$ is minimally unsatisfiable (also $D_{m-1}$ fulfils this), and second that $T(v)$ is a hitting clause-set, that is, every pair of different clauses clashes in at least one variable. These two requirements ensure that the conflict structure of the original (non-boolean) clause-set is preserved by the (boolean) translation. Instead of using $H_{m-1}$ one could actually use any unsatisfiable hitting clause-set with $m$ clauses here.
\item The \emph{simple logarithmic translation}\footnote{called the ``log encoding'' (for CSP-problems) in \cite{Pre09HBSAT}} considers the smallest natural number $p$ with $2^p \ge m$, and sets $T(v) = A_p$, where $A_p$ consists of all $2^p$ full clauses over variables $v_1, \dots, v_p$, while $\gamma_v$ is an arbitrary injection.\footnote{If $2^p = m$, then there is (essentially) only one choice for $\gamma_v$, however otherwise the situation is more complicated, and also resolutions are possible between the remaining clauses, shortening these clauses, and these shortened clauses can be used to shorten the main clauses. Therefore we speak of the ``simple'' translation, and further investigations are needed to find stronger schemes when $m < 2^p$.}
\end{enumerate}

With the exception of the direct translation, which is fully symmetric in the clauses $\gamma_v(\ve)$, one has to decide about the choice $\gamma_v$ of necessary clauses. With the exception of the simple logarithmic translation this is the choice of a suitable bijection, i.e., a question of ordering the values of the variables. In this initial study we have chosen a ``standard ordering'', with the aim of minimising the size of the clause-set, by simply assigning the larger clauses to the larger $k$-values (since the larger the size of arithmetic progressions the fewer there are). Considering our running example $\Fgt{3,3,3}{5}$ we obtain the following $7$ translations:
\begin{enumerate}
\item For the direct encoding we get $5 \cdot 3 = 15$ boolean variables $v_{p,i}$ for $p \in \set{2,3,5,7}$ and $i \in \set{1,2,3}$. The clause $\set{(3,i),(5,i),(7,i)}$ is replaced by $\set{v_{3,i},v_{5,i},v_{7,i}}$ for $i \in \set{1,2,3}$, while clause $\set{(3,i),(7,i),(11,i)}$ is replaced by $\set{v_{3,i},v_{7,i},v_{11,i}}$. For the weak translation we have the $5$ additional clauses $\set{\ol{v_{p,1}},\ol{v_{p,2}},\ol{v_{p,3}}}$ for $p \in \set{2,3,5,7,11}$, while for the strong translation additionally we have the $5 \cdot \binom 32 = 15$ binary clauses $\set{v_{p,i},v_{p,j}}$ for $p \in \set{2,3,5,7,11}$ and $i, j \in \set{1,2,3}$, $i < j$.
\item For the reduced encoding we get $5 \cdot 2 = 10$ boolean variables $v_{p,i}$ for $p \in \set{2,3,5,7,11}$ and $i \in \set{1,2}$. The clause $\set{(3,i),(5,i),(7,i)}$ is replaced by $\set{v_{3,i},v_{5,i},v_{7,i}}$ for $i \in \set{1,2}$ resp.\ by $\set{\ol{v_{3,1}},\ol{v_{3,2}},\ol{v_{5,1}},\ol{v_{5,2}},\ol{v_{7,1}},\ol{v_{7,2}}}$ for $i=3$, while clause $\set{(3,i),(7,i),(11,i)}$ is replaced by $\set{v_{3,i},v_{7,i},v_{11,i}}$ for $i \in \set{1,2}$ resp.\ by $\set{\ol{v_{3,1}},\ol{v_{3,2}},\ol{v_{7,1}},\ol{v_{7,2}},\ol{v_{11,1}},\ol{v_{11,2}}}$ for $i=3$. For the weak translation there are no additional clauses, while for the strong translation we have $5 \cdot \binom 22 = 5$ additional binary clauses $\set{v_{p,1},v_{p,2}}$ for $p \in \set{2,3,5,7,11}$.

  Note that due to our standardisation scheme the long replacement-clause is uniformly used for $i=3$, while actually for each of the five (non-boolean) variables $2,3,5,7,11$ one could use a different $i \in \set{1,2,3}$.
\item For the nested encoding we also get $5 \cdot 2 = 10$ boolean variables $v_{p,i}$ for $p \in \set{2,3,5,7,11}$ and $i \in \set{1,2}$. The clause $\set{(3,i),(5,i),(7,i)}$ is replaced for $i=1,2,3$ by respectively $\set{v_{3,1},v_{5,1},v_{7,1}}$, $\set{\ol{v_{3,1}},v_{3,2},\ol{v_{5,1}},v_{5,2},\ol{v_{7,1}},v_{7,2}}$, $\set{\ol{v_{3,1}},\ol{v_{3,2}},\ol{v_{5,1}},\ol{v_{5,2}},\ol{v_{7,1}},\ol{v_{7,2}}}$, while clause $\set{(3,i),(7,i),(11,i)}$ for $i=1,2,3$ is replaced by respectively $\set{v_{3,1},v_{7,1},v_{11,1}}$, $\set{\ol{v_{3,1}},v_{3,2},\ol{v_{7,1}},v_{7,2},\ol{v_{11,1}},v_{11,2}}$, $\set{\ol{v_{3,1}},\ol{v_{3,2}},\ol{v_{7,1}},\ol{v_{7,2}},\ol{v_{11,1}},\ol{v_{11,2}}}$. For the weak translation there are no additional clauses, while for the strong translation we have $5 \cdot \binom 22 = 5$ additional binary clauses $\set{v_{p,1},v_{p,2}}$ for $p \in \set{2,3,5,7,11}$.

  Note (again) that due to our standardisation scheme the order of the three replacement-clauses is fixed for each variable, while for each variable one could use one of the $3! = 6$ possible orders.
\item Finally, for the logarithmic encoding we get (again, but here this is just an exception) $5 \cdot 2 = 10$ boolean variables $v_{p,i}$ for $p \in \set{2,3,5,7,11}$ and $i \in \set{1,2}$. We use the order $A_2 = \setb{\set{v_1,v_2}, \set{\ol{v_1},v_2}, \set{\ol{v_1},\ol{v_2}}, \set{v_1,\ol{v_2}}}$, where the first three clauses are used for the values $i=1,2,3$. Then the clause $\set{(3,i),(5,i),(7,i)}$ is replaced for $i=1,2,3$ by $\set{v_{3,1},v_{3,2},v_{5,1},v_{5,2},v_{7,1},v_{7,2}}$, $\set{\ol{v_{3,1}},v_{3,2},\ol{v_{5,1}},v_{5,2},\ol{v_{7,1}},v_{7,2}}$, $\set{\ol{v_{3,1}},\ol{v_{3,2}},\ol{v_{5,1}},\ol{v_{5,2}},\ol{v_{7,1}},\ol{v_{7,2}}}$ respectively, and $\set{(3,i),(7,i),(11,i)}$ is\, replaced \, resp.\ \, by \, $\set{v_{3,1},v_{3,2},v_{7,1},v_{7,2},v_{11,1},v_{11,2}}$, \\$\set{\ol{v_{3,1}},v_{3,2},\ol{v_{7,1}},v_{7,2},\ol{v_{11,1}},v_{11,2}}$, $\set{\ol{v_{3,1}},\ol{v_{3,2}},\ol{v_{7,1}},\ol{v_{7,2}},\ol{v_{11,1}},\ol{v_{11,2}}}$. Addi\-ti\-onal\-ly we have the $5$ clauses $\set{v_{p,1},\ol{v_{p,2}}}$ for $p \in \set{2,3,5,7,11}$.
\end{enumerate}

Somewhat surprisingly, in many cases considered in this paper the weak nested translation turned out to be best (from the above 7 translations considered), for all three types of solvers, look-ahead, conflict-driven and local-search solvers (where for the latter an appropriate algorithm has to be chosen). Only for larger number of colours is the logarithmic translation superior (for local search, with various best algorithms; complete solvers were not successful on any of these instances (with larger number of colours)), while in all cases the weak nested translation was superior over the direct translation (weak or strong, for all solver types).

\section{Computing Green-Tao numbers}
\label{sec:compGreenTao}

For trivial GT-numbers as with vdW-numbers we have $\gtz m2 = m+1$. However the simple GT-numbers are non-trivial: $\gtz 1k$ is the smallest $n$ such that the first $n$ prime numbers contain an arithmetic progression of size $k$. Only the values for $2 \le k \le 21$ are known, given by the sequence $2,4,9,10,37,155,263,289,316,\\ 21'966, 23'060, 58'464, 2'253'121, 9'686'320, 11'015'837, 227'225'515, 755'752'809,\\ 3'466'256'932,  22'009'064'470, 220'525'414'079$.\footnote{This data is available at \url{http://users.cybercity.dk/~dsl522332/math/aprecords.htm}, in the form of the prime numbers themselves, not their indices (as used by us), and so we needed to rank the prime numbers there.} It seems likely that consideration of GT-numbers for core tuples involving $k \ge 11$ is infeasible (since the first $21966$ prime numbers need to be considered just to see the \emph{first} progression of size $11$).

Considering the generalised clause-sets $F = \Fgt{k_1,\dots,k_m}{n}$ (recall Section \ref{sec:transnbb}), the number $n(F)$ of (formal) variables is $n$, while the number $c(F)$ of clauses is $\sum_{i=1}^m \abs{E(\arithpp(k_i,n))}$.\footnote{It seems that for $n$ not much smaller than $\gtz{m}{k_1,\dots,k_m}$ all variables actually occur with the exception of prime number $2$, which occurs iff some $k_i = 2$ exists.} So to compute the number of clauses in $F$, we have to compute how many arithmetic progressions of size $k$ there are for a given $n$. In other words, what can be said about the number $\abs{E(\arithpp(k,n))}$ of hyperedges in the GT-hypergraphs? Exploiting the famous (unproven) ``$m$-tuples conjecture'' of Hardy-Littlewood, various asymptotic formulas (where the quotient with the true value is approaching $1$ with $n$ going to infinity) are given in \cite{GrosswaldHagis1979ProgressionsPrimes}. Translated into our context, where we rank the primes, for arbitrary $N \in \NNZ$ the formula (7) from \cite{GrosswaldHagis1979ProgressionsPrimes} yields, using $x := n \cdot \log n$ (which is an asymptotically precise formula to translate from the rank $n$ to the associate $n$-th prime number $p_n$):
\begin{displaymath}
  \abs{E(\arithpp(k,n))} \sim C_k \cdot \frac{x^2}{(\log x)^k} \cdot (1 + \sum_{i=1}^N \frac{a_{k,i}}{(\log x)^i})
\end{displaymath}
(for fixed $k$, proven meanwhile for $k \le 4$; precise formulas for $C_k$ and the $a_i$ are also given in \cite{GrosswaldHagis1979ProgressionsPrimes}). Just using linear regression to determine $C_k$ and the $a_{k,i}$, using $N=2$, yields very good approximations over the ranges we are considering.

Solvers used are the algorithms from the \texttt{Ubcsat} local-search suite (\cite{TompkinsHoos2004Ubcsat}), \texttt{minisat2} (\cite{EenSoerensson2003Minisat}) for conflict-driven solvers (on our instances either \texttt{minisat2} was superior or not much worse than all other publicly available conflict-driven solvers, and thus it seems that the optimisations applied to \texttt{minisat2} in other solvers don't improve performance on our instances), and \texttt{OKsolver-2002} (\cite{Ku2002h}), \texttt{march\_pl} (\cite{Heule2008PhD}) and \texttt{satz} (\cite{Li1999Satz}) for look-ahead solvers. In one (largest) case \texttt{survey propagation} (\cite{BraunsteinMezardZecchina2005SurveyProp}) was successful (with $708206$ clauses of length $5$). If not stated otherwise, for all non-boolean cases the weak (standard) nested translation is best (considering complete and incomplete solvers), and if not otherwise stated, for lower bounds \texttt{rnovelty+} is best. Recall that a lower bound stated as ``$\ge n$'' means that we conjecture that actually equality holds.

We were able to compute five core GT-numbers, for $3$ core numbers we have reasonable conjectures, and for $9$ core numbers we have hopefully not unreasonable lower bounds. Furthermore we were able to compute $12$ extended core GT-numbers, while for $16$ cases we have conjectures. Transversal GT-numbers are presented in Section \ref{sec:transgtnumbers}; see Subsection \ref{sec:remarkstrans} for general remarks.
\begin{enumerate}
\item $4$ binary core GT-numbers $\gtz{2}{a,b}$ are known:
  \begin{center}
    \begin{tabular}[c]{c||*{5}{c@{ $\;$}}}
      $\begin{array}[r]{c@{\hspace{0.5em}}r}
        & b\\[-1.9ex]
        a
      \end{array}\hspace{-0.1em}$ & $3$ & $4$ & $5$ & $6$ & $7$\\
      \hline\hline
      $3$ & $23$ & $79$ & $528$ & $\ge 2072$ & $> 13800$\\
      $4$ & - & 512 & $> 4231$\\
      $5$ & - & - & $\ge 34309$
    \end{tabular}\\[1ex]
    For $(5,5)$ we experienced the only case where \texttt{survey propagation} was successful (converging for $n < 34309$, diverging for $n \ge 34309$). For the other lower bounds \texttt{adaptnovelty+} is best. \texttt{OKsolver-2002} is best for $(4,4)$, while for $(3,5)$ \texttt{minisat2} is best, followed by \texttt{march\_pl}.
  \end{center}
\item One ternary core GT-number $\gtz{3}{a,b,c}$ is known:
  \begin{center}
    \begin{tabular}[c]{c||*{3}{c}}
      $\begin{array}[r]{c@{\hspace{0.5em}}r}
        & c\\[-1.9ex]
        a,b
      \end{array}\hspace{-0.1em}$ & $3$ & $4$ & $5$\\
      \hline\hline
      $3,3$ & $137$ & $\ge 434$ & $> 1989$\\
      $3,4$ & - & $> 1662$ & $> 8300$\\
      $4,4$ & - & $> 5044$
    \end{tabular}\\[1ex]
    For $(3,3,3)$ the logarithmic translation performed best, with \texttt{minisat2} fastest, followed by \texttt{OKsolver-2002}. For $(3,4,5)$ \texttt{rnovelty} performed best.
  \end{center}
\item No core GT-number $\gtz{4}{a,b,c,d}$ is known:
  \begin{center}
    \begin{tabular}[c]{c||*{2}{c}}
      $\begin{array}[r]{c@{\hspace{0.5em}}r}
        & d\\[-1.9ex]
        a,b,c
      \end{array}\hspace{-0.1em}$ & $3$ & $4$\\
      \hline\hline
      $3,3,3$ & $> 384$ & $> 1052$\\
      $3,3,4$ & - & $> 2750$
    \end{tabular}
  \end{center}
\item Extending $(3,3)$ by $m$ $2$'s, i.e., the numbers $\gtz{m+2}{2, \dots, 2, 3,3}$:
  \begin{center}
    \hspace*{-2em}
    \begin{tabular}[c]{c||c|*{14}{c@{ $\;$}}}
      $m$ & $0$ & $1$ & $2$ & $3$ & $4$ & $5$ & $6$ & $7$ & $8$ & $9$ & $10$ & $11$ & $12$ & $13$ & $14$\\
      \hline\hline
      & 23 & 31 & 39 & 41 & 47 & $53$ & $55$ & $\ge 60$ & $\ge 62$ & $\ge 67$ & $\ge 71$ & $\ge 73$ & $\ge 82$ & $\ge 83$ & $\ge 86$
    \end{tabular}\\[1ex]
    \texttt{minisat2} is the best complete solver here (also for the other (complete) cases below). For the lower bounds the logarithmic translation is best, with \texttt{rsaps} except for $m=13$ where \texttt{walksat-tabu} without null-flips is best.
  \end{center}
\item Extending $(3,k)$ for $k \ge 4$ by $m$ $2$'s, i.e., the numbers $\gtz{m+2}{2, \dots, 2, 3,k}$:
  \begin{center}
    \hspace*{-2em}
    \begin{tabular}[c]{c||c|*{7}{c@{ $\;$}}}
      $\begin{array}[r]{c@{\hspace{0.5em}}r}
        & m\\[-1.9ex]
        k
      \end{array}\hspace{-0.1em}$ & $0$ & $1$ & $2$ & $3$ & $4$ & $5$ & $6$\\
      \hline\hline
      $4$ & 79 & 117 & 120 & 128 & $136$ & $\ge 142$ & $\ge 151$\\
      $5$ & 528 & $581$ & $\ge 582$ & $\ge 610$
    \end{tabular}\\[1ex]
    For $k=5$, $m=2$ \texttt{saps} is best, and for $m=3$ \texttt{walksat}. For $k=4$, $m=6$ \texttt{walksat-tabu} with the logarithmic translation is best.
  \end{center}
\item Extending $(4,k)$ by $m$ $2$'s, i.e., the numbers $\gtz{m+2}{2, \dots, 2, 4,k}$:
  \begin{center}
    \begin{tabular}[c]{c||c|*{3}{c@{ $\;$}}}
      $\begin{array}[r]{c@{\hspace{0.5em}}r}
        & m\\[-1.9ex]
        k
      \end{array}\hspace{-0.1em}$ & $0$ & $1$ & $2$ \\
      \hline\hline
      $4$ & 512 & $\ge 553$ & $> 588$
    \end{tabular}\\[1ex]
    \texttt{sapsnr} is best (for the lower bounds).
  \end{center}
\item Extending $(3,3,k)$ by $m$ $2$'s, i.e., the numbers $\gtz{m+3}{2, \dots, 2, 3,3,k}$:
  \begin{center}
    \begin{tabular}[c]{c||c|*{4}{c@{ $\;$}}}
      $\begin{array}[r]{c@{\hspace{0.5em}}r}
        & m\\[-1.9ex]
        k
      \end{array}\hspace{-0.1em}$ & $0$ & $1$ & $2$ \\
      \hline\hline
      $3$ & 137 & $151$ & $\ge 154$\\
      $4$ & $\ge 434$ & $\ge 453$ & $> 471$
    \end{tabular}
  \end{center}
\end{enumerate}

Some final remarks:
\begin{enumerate}
\item For vdW-numbers, Conjecture \ref{conj:onecomp} generalised to all core tuples says that in every row of a table of core numbers (that is, in a core tuple one component grows while the others are fixed) we would have growth-rates $O(n^k)$ (where $k$ depends on the row). Now for GT-numbers a first guess is that we have growth-rates $O(\exp(n^k))$.
\item For vdW-numbers, in \cite{LandmanRobertson2003ArithmeticProgressions}, Research Problem 2.8.6, it is conjectured that $\waez{2}{k,k} \ge \waez{2}{k-1,k+1} \ge \waez{2}{k-2,k+2} \ge \dots \ge \waez{2}{2,2 k - 2}$. The sequences $N(k,k), N(k-1,k+1), \dots, N(2,2k+2)$ can for $N = \waez{2}{}$ be reasonably evaluated for $2 \le k \le 6$, yielding the sequences 
  \begin{displaymath}
    (3), (9,7), (35,22,11), (178,73,46,15), (1132,\text{?},146,77,19),
  \end{displaymath}
  which supports the conjecture. For GT-numbers (that is, $N = \gtz{2}{}$) we can reasonably evaluate $2 \le k \le 4$, obtaining the sequences
  \begin{displaymath}
    (3), (23,14), (79,528,55),
  \end{displaymath}
  and we see that now we have a more complicated behaviour.
\end{enumerate}

\section{Open problems and outlook}
\label{sec:summary}

Regarding the generic translation scheme, further extensive experimentation is needed w.r.t.\ the problem of ordering the values and of mixing translation schemes (recall that every variable can be treated on its own). Also further instances of the generic scheme need to considered, starting with refining the logarithmic translation when the number of values is not a power of $2$. Of course, finally some form of understanding needs to be established, and we hope that the generic scheme offers a suitable environment for such considerations.

As mentioned in Subsection \ref{sec:remarkstrans}, the translation of transversal extension problems into boolean SAT can use cardinality constraints, and this needs to be explored systematically. This includes the special case of transversal extensions of simple tuples, which is basically the hypergraph transversal problem (for these special hypergraphs).

A fundamental problem is to improve performance on \emph{unsatisfiable} instances (of complete solvers). The most promising general approach seems to us to systematically study the optimisation of heuristics as outlined in \cite{Kullmann2007HandbuchTau}. Investigating the tree-resolution and full-resolution complexity of these instances should be of great interest; we noticed that especially with the \texttt{OKsolver-2002} the search trees show remarkable regularities (of a number-theoretical touch, in a kind of ``fractal'' way). Exploiting the monotone nature of the hypergraph sequences of vdW- or GT-hypergraphs seems also necessary to reach the next level of vdW- or GT-numbers (regarding core parameter tuples), where some first (sporadic) methods one finds in \cite{KourilPaulW26}.

In general, it seems to us that instances from Ramsey theory, like vdW-instances or GT-instances as considered in this paper, or like the Ramsey-instances (and there are many other families), provide very good benchmarks for SAT solvers, combining the power of systematic creation as for random instances with various types of ``structures'', where the interplay between these structures and SAT solving should be of great interest and potential.

\begin{appendix}

\section{Transversal numbers}
\label{sec:Transversalnumbers}

A generic way of computing transversal vdW-numbers and transversal GT-numbers via SAT-solvers is as follows, using $N \in \set{\waez{}{}, \gtz{}{}}$ and respectively $G_k(n) = \arithp(k,n)$ or $G_k(n) = \arithpp(k,n)$.
\begin{itemize}
\item Let $\tau_k(n) := \tau(G_k(n))$.
\item Recall that $N_{m+1}(2,\dots,2,k)$ is the smallest $n$ with $\tau_k(n) > m$.
\item So we compute the numbers $\tau_k(n)$ for $n = 1, 2, \dots$ as far as we get, and derive from these transversal numbers the transversal $N$-numbers.
\item We start with $\tau_k(1) = 0$ (for $k>1$), and we know 
  \begin{displaymath}
    \tau_k(n+1) \in \set{\tau_k(n),\tau_k(n)+1}.
  \end{displaymath}
\item So for computing $\tau_k(n+1)$ we consider the satisfiability problem 
  \begin{displaymath}
    \tau_k(n+1) = b \text{ ?}
  \end{displaymath}
for $b := \tau_k(n)$: If this problem is satisfiable (a transversal for $G_k(n+1)$ of size $b$ exists), then we have $\tau_k(n+1) = b$, while otherwise we have $\tau_k(n+1) = b+1$.
\item To formulate the satisfiability problem as a (boolean) CNF-SAT-problem, we consider the vertices of $G_k(n+1)$ as variables $v_1, \dots, v_{n+1}$ and the hyperedges as positive clauses, and we add clauses expressing the cardinality constraint ``$v_1 + \dots + v_{n+1} = b$'' (considering $v_i \in \set{0,1} \subset \NNZ$).
\end{itemize}
The best combination of SAT-solver and cardinality-constraint-translation we found uses \texttt{minisat2} and binary addition. Except for $\waez{m+1}{2,\dots,2,k}$ all numbers are computed in this way.

\subsection{Transversal vdW-numbers}
\label{sec:transvdwnumbers}

Now we present the transversal vdW-numbers $\waez{m+1}{2,\dots,2,k}$ we have computed. Numbers in boldface are given by the (precise) formula in \cite{LandmanRobertsonCulver2005vanderWaerden}, while for underlined numbers the formula still holds though the $m$-value is not in the domain of proven correctness.

\begin{enumerate}
\item The known numbers $\waez{m+1}{2,\dots,2,3}$, with $m = a \cdot 10 + b$:
    \begin{center}
    \begin{tabular}[c]{c||*{10}{c}}
      $\begin{array}[r]{c@{\hspace{0.5em}}r}
        & b\\[-1.9ex]
        a
      \end{array}$ & $0$ & $1$ & $2$ & $3$ & $4$ & $5$ & $6$ & $7$ & $8$ & $9$\\
      \hline\hline
      $0$ & \textbf{3} & \textbf{6} & \ul{7} & \ul{8} & 10 & 12 & 15 & 16 & 17 & 18\\
      $1$ & 19 & 21 & 22 & 23 & 25 & 27 & 28 & 29 & 31 & 33\\
      $2$ & 34 & 35 & 37 & 38 & 39 & 42 & 43 & 44 & 45 & 46\\
      $3$ & 47 & 48 & 49 & 50 & 52 & 53 & 55 & 56 & 57 & 59\\
      $4$ & 60 & 61 & 62 & 64 & 65 & 66 & 67 & 68 & 69 & 70\\
      $5$ & 72 & 73 & 75 & 76 & 77 & 78 & 79 & 80 & 81 & 83\\
      $6$ & 85 & 86 & 87 & 88 & 89 & 90 & 91 & 93 & 94 & 96\\
      $7$ & 97 & 98 & 99 & 101 & 102 & 103 & 105 & 106 & 107 & 108\\
      $8$ & 109 & 110 & 112 & 113 & 115 & 116 & 117 & 118 & 119 & 120\\
      $9$ & 123 & 124 & 125 & 126 & 127 & 128 & 129 & 130 & 131 & 132\\
      $10$ & 133 & 134 & 135 & 136 & 138 & 139 & 140 & 141 & 142 & 143\\
      $11$ & 144 & 146 & 147 & 148 & 149 & 151 & 152 & 153 & 154 & 155\\
      $12$ & 156 & 158 & 159 & 160 & 161 & 162 & 164 & 166 & 167 & 168\\
      $13$ & 170 & 171 & 172 & 173 & 175 & 176 & 177 & 178 & 179 & 180\\
      $14$ & 181 & 182 & 183 & 184 & 185 & 186 & 187 & 188 & 189 & 190\\
      $15$ & 191 & 192 & 193
    \end{tabular}
  \end{center}
  Here via SAT solving only up to $\tau(\arithp(3,101)) = 74$, $\tau(\arithp(3,102)) = 75$ could be computed (yielding $\waez{74+1}{2,\dots,2,3} = 102$), while the data for $m > 74$ has been collected by Jarek Wroblewski at \url{http://www.math.uni.wroc.pl/~jwr/non-ave.htm}, using special methods. The data is available in the form of an ``$\alpha$-steplist'', that is, for index $i = 1, 2, 3, \dots$ the smallest $a_i = n \in \NN$ with $\alpha(\arithp(3,n)) = i$ is given. Jarek Wroblewski conjectures that $a_i \le i^{1.5}$. This data for $i = b \cdot 10 + c$ is given by the following table.
    \begin{center}
    \begin{tabular}[c]{c||*{10}{c}}
      $\begin{array}[r]{c@{\hspace{0.5em}}r}
        & c\\[-1.9ex]
        b
      \end{array}$ & $1$ & $2$ & $3$ & $4$ & $5$ & $6$ & $7$ & $8$ & $9$ & $10$\\
      \hline\hline
      $0$ & 1 & 2 & 4 & 5 & 9 & 11 & 13 & 14 & 20 & 24\\
      $1$ & 26 & 30 & 32 & 36 & 40 & 41 & 51 & 54 & 58 & 63\\
      $2$ & 71 & 74 & 82 & 84 & 92 & 95 & 100 & 104 & 111 & 114\\
      $3$ & 121 & 122 & 137 & 145 & 150 & 157 & 163 & 165 & 169 & 174\\
      $4$ & 194
    \end{tabular}
  \end{center}
  We would get such sequence of these numbers also directly from the runs of the SAT-solvers (as discussed above, however only as far as we get), by collecting all the $n$ for which a satisfiable instance was obtained (which means that the transversal number didn't change, which is equivalent to the independence numbers making a step (of $+1$)).
\item The known numbers $\waez{m+1}{2,\dots,2,4}$, with $m = a \cdot 10 + b$:
  \begin{center}
    \begin{tabular}[c]{c||*{10}{c}}
      $\begin{array}[r]{c@{\hspace{0.5em}}r}
        & b\\[-1.9ex]
        a
      \end{array}$ & $0$ & $1$ & $2$ & $3$ & $4$ & $5$ & $6$ & $7$ & $8$ & $9$\\
      \hline\hline
      $0$ & \textbf{4} & \textbf{7} & \textbf{11} & 12 & 14 & 16 & 18 & 20 & 22 & 24\\
      $1$ & 26 & 29 & 31 & 32 & 35 & 36 & 38 & 39 & 41 & 42\\
      $2$ & 44 & 46 & 47 & 49 & 51 & 52 & 55 & 56 & 57 & 59\\
      $3$ & 61 & 62 & 63 & 65 & 67 & 69 & 71 & 72 & 73 & 75\\
      $4$ & 76 & 78 & 80
    \end{tabular}
  \end{center}
\item The known numbers $\waez{m+1}{2,\dots,2,5}$, with $m = a \cdot 10 + b$:
  \begin{center}
    \begin{tabular}[c]{c||*{10}{c}}
      $\begin{array}[r]{c@{\hspace{0.5em}}r}
        & b\\[-1.9ex]
        a
      \end{array}$ & $0$ & $1$ & $2$ & $3$ & $4$ & $5$ & $6$ & $7$ & $8$ & $9$\\
      \hline\hline
      $0$ & \textbf{5} & \textbf{10} & \textbf{15} & \textbf{20} & \ul{21} & 22 & 23 & 26 & 30 & 32\\
      $1$ & 35 & 40 & 45 & 46 & 47 & 48 & 50 & 53 & 55 & 60\\
      $2$ & 65 & 70 & 71 & 72 & 73 & 74 & 75 & 80 & 85 & 90\\
      $3$ & 95 & 96 & 97 & 98 & 99 & 100 & 101 & 102 & 103
    \end{tabular}
  \end{center}
\item The known numbers $\waez{m+1}{2,\dots,2,6}$, with $m = a \cdot 10 + b$:
  \begin{center}
    \begin{tabular}[c]{c||*{10}{c}}
      $\begin{array}[r]{c@{\hspace{0.5em}}r}
        & b\\[-1.9ex]
        a
      \end{array}$ & $0$ & $1$ & $2$ & $3$ & $4$ & $5$ & $6$ & $7$ & $8$ & $9$\\
      \hline\hline
      $0$ & \textbf{6} & \textbf{11} & \textbf{16} & \textbf{21} & \textbf{27} & 28 & 30 & 31 & 34 & 38\\
      $1$ & 42 & 43 & 47 & 52 & 53 & 55 & 57 & 60 & 63 & 67\\
      $2$ & 69 & 72 & 77 & 78 & 79 & 81 & 84
    \end{tabular}
  \end{center}
\item The known numbers $\waez{m+1}{2,\dots,2,7}$, with $m = a \cdot 10 + b$:
  \begin{center}
    \begin{tabular}[c]{c||*{10}{c}}
      $\begin{array}[r]{c@{\hspace{0.5em}}r}
        & b\\[-1.9ex]
        a
      \end{array}$ & $0$ & $1$ & $2$ & $3$ & $4$ & $5$ & $6$ & $7$ & $8$ & $9$\\
      \hline\hline
      $0$ & \textbf{7} & \textbf{14} & \textbf{21} & \textbf{28} & \textbf{35} & \textbf{42} & \ul{43} & 44 & 45 & 47\\
      $1$ & 49 & 54 & 58 & 62 & 66 & 70 & 77 & 84 & 91 & 92\\
      $2$ & 93 & 94 & 96 & 97 & 99 & 105 & 108
    \end{tabular}
  \end{center}
\item The known numbers $\waez{m+1}{2,\dots,2,8}$, with $m = a \cdot 10 + b$:
  \begin{center}
    \begin{tabular}[c]{c||*{10}{c}}
      $\begin{array}[r]{c@{\hspace{0.5em}}r}
        & b\\[-1.9ex]
        a
      \end{array}$ & $0$ & $1$ & $2$ & $3$ & $4$ & $5$ & $6$ & $7$ & $8$ & $9$\\
      \hline\hline
      $0$ & \textbf{8} & \textbf{15} & \textbf{22} & \textbf{29} & \textbf{36} & \textbf{43} & \textbf{51} & 52 & 53 & 55\\
      $1$ & 57 & 60 & 64 & 70 & 73 & 79 & 81 & 86 & 93 & 100\\
      $2$ & 101 & 102 & 103
    \end{tabular}
  \end{center}
\item The known numbers $\waez{m+1}{2,\dots,2,9}$, with $m = a \cdot 10 + b$:
  \begin{center}
    \begin{tabular}[c]{c||*{10}{c}}
      $\begin{array}[r]{c@{\hspace{0.5em}}r}
        & b\\[-1.9ex]
        a
      \end{array}$ & $0$ & $1$ & $2$ & $3$ & $4$ & $5$ & $6$ & $7$ & $8$ & $9$\\
      \hline\hline
      $0$ & \textbf{9} & \textbf{18} & \textbf{25} & \textbf{32} & \textbf{39} & \textbf{46} & \ul{53} & 58 & 59 & 62\\
      $1$ & 66 & 72 & 74 & 77 & 81 & 87 & 91 & 97 & 102 & 106\\
      $2$ & 110
    \end{tabular}
  \end{center}
\item The known numbers $\waez{m+1}{2,\dots,2,10}$, with $m = a \cdot 10 + b$:
  \begin{center}
    \begin{tabular}[c]{c||*{10}{c}}
      $\begin{array}[r]{c@{\hspace{0.5em}}r}
        & b\\[-1.9ex]
        a
      \end{array}$ & $0$ & $1$ & $2$ & $3$ & $4$ & $5$ & $6$ & $7$ & $8$ & $9$\\
      \hline\hline
      $0$ & \textbf{10} & \textbf{19} & \textbf{29} & \ul{34} & \ul{41} & \ul{48} & \ul{55} & 62 & 65 & 69\\
      $1$ & 74 & 79 & 85 & 89 & 92 & 96 & 101 & 106 & 110
    \end{tabular}
  \end{center}
\item The known numbers $\waez{m+1}{2,\dots,2,11}$, with $m = a \cdot 10 + b$:
  \begin{center}
    \begin{tabular}[c]{c||*{10}{c}}
      $\begin{array}[r]{c@{\hspace{0.5em}}r}
        & b\\[-1.9ex]
        a
      \end{array}$ & $0$ & $1$ & $2$ & $3$ & $4$ & $5$ & $6$ & $7$ & $8$ & $9$\\
      \hline\hline
      $0$ & \textbf{11} & \textbf{22} & \textbf{33} & \textbf{44} & \textbf{55} & \textbf{66} & \textbf{77} & \textbf{88} & \textbf{99} & \textbf{110}\\
      $1$ & \ul{111} & 112 & 113 & 114 & 116 & 118 & 119 & 121 & 129
    \end{tabular}
  \end{center}
\item The known numbers $\waez{m+1}{2,\dots,2,12}$, with $m = a \cdot 10 + b$:
  \begin{center}
    \begin{tabular}[c]{c||*{10}{c}}
      $\begin{array}[r]{c@{\hspace{0.5em}}r}
        & b\\[-1.9ex]
        a
      \end{array}$ & $0$ & $1$ & $2$ & $3$ & $4$ & $5$ & $6$ & $7$ & $8$ & $9$\\
      \hline\hline
      $0$ & \textbf{12} & \textbf{23} & \textbf{34} & \textbf{45} & \textbf{56} & \textbf{67} & \textbf{78} & \textbf{89} & \textbf{100} & \textbf{111}\\
      $1$ & \textbf{123} & 124 & 125 & 126 & 127 & 129 & 130 & $133$
    \end{tabular}
  \end{center}
\item The known numbers $\waez{m+1}{2,\dots,2,13}$, with $m = a \cdot 10 + b$:
  \begin{center}
    \begin{tabular}[c]{c||*{10}{c}}
      $\begin{array}[r]{c@{\hspace{0.5em}}r}
        & b\\[-1.9ex]
        a
      \end{array}$ & $0$ & $1$ & $2$ & $3$ & $4$ & $5$ & $6$ & $7$ & $8$ & $9$\\
      \hline\hline
      $0$ & \textbf{13} & \textbf{26} & \textbf{39} & \textbf{52} & \textbf{65} & \textbf{78} & \textbf{91} & \textbf{104} & \textbf{117} & \textbf{130}\\
      $1$ & \textbf{143} & \textbf{156} & \ul{157} & 158 & 159 & 160 & 162 & $163$
    \end{tabular}
  \end{center}
\item The known numbers $\waez{m+1}{2,\dots,2,14}$, with $m = a \cdot 10 + b$:
  \begin{center}
    \begin{tabular}[c]{c||*{10}{c}}
      $\begin{array}[r]{c@{\hspace{0.5em}}r}
        & b\\[-1.9ex]
        a
      \end{array}$ & $0$ & $1$ & $2$ & $3$ & $4$ & $5$ & $6$ & $7$ & $8$ & $9$\\
      \hline\hline
      $0$ & \textbf{14} & \textbf{27} & \textbf{40} & \textbf{53} & \textbf{66} & \textbf{79} & \textbf{92} & \textbf{105} & \textbf{118} & \textbf{131}\\
      $1$ & \textbf{144} & \textbf{157} & \textbf{171} & 172 & 173 & 174 & 175 & 176
    \end{tabular}
  \end{center}
\item The known numbers $\waez{m+1}{2,\dots,2,15}$, with $m = a \cdot 10 + b$:
  \begin{center}
    \begin{tabular}[c]{c||*{10}{c}}
      $\begin{array}[r]{c@{\hspace{0.5em}}r}
        & b\\[-1.9ex]
        a
      \end{array}$ & $0$ & $1$ & $2$ & $3$ & $4$ & $5$ & $6$ & $7$ & $8$ & $9$\\
      \hline\hline
      $0$ & \textbf{15} & \textbf{30} & \textbf{43} & \textbf{56} & \textbf{69} & \textbf{82} & \textbf{95} & \textbf{108} & \textbf{121} & \ul{134}\\
      $1$ & \ul{147} & \ul{160} & \ul{173} & 184 & 185 & 186 & 188 & 189
    \end{tabular}
  \end{center}
\item The known numbers $\waez{m+1}{2,\dots,2,16}$, with $m = a \cdot 10 + b$:
  \begin{center}
    \begin{tabular}[c]{c||*{10}{c}}
      $\begin{array}[r]{c@{\hspace{0.5em}}r}
        & b\\[-1.9ex]
        a
      \end{array}$ & $0$ & $1$ & $2$ & $3$ & $4$ & $5$ & $6$ & $7$ & $8$ & $9$\\
      \hline\hline
      $0$ & \textbf{16} & \textbf{31} & \textbf{47} & \textbf{58} & \ul{71} & \ul{84} & \ul{97} & \ul{110} & \ul{123} & \ul{136}\\
      $1$ & \ul{149} & \ul{162} & \ul{175} & 188 & 197 & 199 & 200 & 202
    \end{tabular}
  \end{center}
\item The known numbers $\waez{m+1}{2,\dots,2,17}$, with $m = a \cdot 10 + b$:
  \begin{center}
    \begin{tabular}[c]{c||*{10}{c}}
      $\begin{array}[r]{c@{\hspace{0.5em}}r}
        & b\\[-1.9ex]
        a
      \end{array}$ & $0$ & $1$ & $2$ & $3$ & $4$ & $5$ & $6$ & $7$ & $8$ & $9$\\
      \hline\hline
      $0$ & \textbf{17} & \textbf{34} & \textbf{51} & \textbf{68} & \textbf{85} & \textbf{102} & \textbf{119} & \textbf{136} & \textbf{153} & \textbf{170}\\
      $1$ & \textbf{187} & \textbf{204} & \textbf{221} & \textbf{238} & \textbf{255} & \textbf{272} & \ul{273} & 274 & 275 & 276\\
      $2$ & 277
    \end{tabular}
  \end{center}
\item The known numbers $\waez{m+1}{2,\dots,2,18}$, with $m = a \cdot 10 + b$:
  \begin{center}
    \begin{tabular}[c]{c||*{10}{c}}
      $\begin{array}[r]{c@{\hspace{0.5em}}r}
        & b\\[-1.9ex]
        a
      \end{array}$ & $0$ & $1$ & $2$ & $3$ & $4$ & $5$ & $6$ & $7$ & $8$ & $9$\\
      \hline\hline
      $0$ & \textbf{18} & \textbf{35} & \textbf{52} & \textbf{69} & \textbf{86} & \textbf{103} & \textbf{120} & \textbf{137} & \textbf{154} & \textbf{171}\\
      $1$ & \textbf{188} & \textbf{205} & \textbf{222} & \textbf{239} & \textbf{256} & \textbf{273} & \textbf{291} & 292 & 293 & 294\\
      $2$ & 295
    \end{tabular}
  \end{center}
\item The known numbers $\waez{m+1}{2,\dots,2,19}$, with $m = a \cdot 10 + b$:
  \begin{center}
    \begin{tabular}[c]{c||*{10}{c}}
      $\begin{array}[r]{c@{\hspace{0.5em}}r}
        & b\\[-1.9ex]
        a
      \end{array}$ & $0$ & $1$ & $2$ & $3$ & $4$ & $5$ & $6$ & $7$ & $8$ & $9$\\
      \hline\hline
      $0$ & \textbf{19} & \textbf{38} & \textbf{57} & \textbf{76} & \textbf{95} & \textbf{114} & \textbf{133} & \textbf{152} & \textbf{171} & \textbf{190}\\
      $1$ & \textbf{209} & \textbf{228} & \textbf{247} & \textbf{266} & \textbf{285} & \textbf{304} & \textbf{323} & \textbf{342} & \ul{343} & 344\\
      $2$ & 345 & 346
    \end{tabular}
  \end{center}
\end{enumerate}

\subsection{Transversal GT-numbers}
\label{sec:transgtnumbers}

\begin{enumerate}
\item The known numbers $\gtz{m+1}{2,\dots,2,3}$, with $m = a \cdot 10 + b$:
    \begin{center}
    \begin{tabular}[c]{c||*{10}{c}}
      $\begin{array}[r]{c@{\hspace{0.5em}}r}
        & b\\[-1.9ex]
        a
      \end{array}$ & $0$ & $1$ & $2$ & $3$ & $4$ & $5$ & $6$ & $7$ & $8$ & $9$\\
      \hline\hline
      $0$ & 4 & 7 & 9 & 13 & 14 & 16 & 18 & 21 & 22 & 23\\
      $1$ & 28 & 29 & 30 & 32 & 33 & 36 & 38 & 39 & 40 & 42\\
      $2$ & 43 & 47 & 48 & 49 & 50 & 52 & 55 & 56 & 57 & 59\\
      $3$ & 61 & 62 & 65 & 68 & 69 & 70 & 71 & 72 & 73 & 75\\
      $4$ & 76 & 78 & 80 & 81
    \end{tabular}
  \end{center}
\item The known numbers $\gtz{m+1}{2,\dots,2,4}$, with $m = a \cdot 10 + b$:
    \begin{center}
    \begin{tabular}[c]{c||*{10}{c}}
      $\begin{array}[r]{c@{\hspace{0.5em}}r}
        & b\\[-1.9ex]
        a
      \end{array}$ & $0$ & $1$ & $2$ & $3$ & $4$ & $5$ & $6$ & $7$ & $8$ & $9$\\
      \hline\hline
      $0$ & 9 & 14 & 17 & 22 & 26 & 32 & 35 & 36 & 37 & 45\\
      $1$ & 46 & 51 & 56 & 58 & 61 & 62 & 71 & 72 & 73 & 78\\
      $2$ & 79 & 83\\
    \end{tabular}
  \end{center}
\item The known numbers $\gtz{m+1}{2,\dots,2,5}$, with $m = a \cdot 10 + b$:
    \begin{center}
    \begin{tabular}[c]{c||*{10}{c}}
      $\begin{array}[r]{c@{\hspace{0.5em}}r}
        & b\\[-1.9ex]
        a
      \end{array}$ & $0$ & $1$ & $2$ & $3$ & $4$ & $5$ & $6$ & $7$ & $8$ & $9$\\
      \hline\hline
      $0$ & 10 & 31 & 32 & 49 & 58 & 61 & 62 & 78 & 87 & 98\\
      $1$ & 107 & 112 & 121 & 123 & 142 & 143
    \end{tabular}
  \end{center}
\item The known numbers $\gtz{m+1}{2,\dots,2,6}$, with $m = a \cdot 10 + b$:
    \begin{center}
    \begin{tabular}[c]{c||*{10}{c}}
      $\begin{array}[r]{c@{\hspace{0.5em}}r}
        & b\\[-1.9ex]
        a
      \end{array}$ & $0$ & $1$ & $2$ & $3$ & $4$ & $5$ & $6$ & $7$ & $8$ & $9$\\
      \hline\hline
      $0$ & 37 & 55 & 64 & 71 & 90 & 97 & 125 & 152 & 162 & 179\\
      $1$ & 201 & 204 & 211 & 212 & 250
    \end{tabular}
  \end{center}
\item The known numbers $\gtz{m+1}{2,\dots,2,7}$, with $m = a \cdot 10 + b$:
    \begin{center}
    \begin{tabular}[c]{c||*{10}{c}}
      $\begin{array}[r]{c@{\hspace{0.5em}}r}
        & b\\[-1.9ex]
        a
      \end{array}$ & $0$ & $1$ & $2$ & $3$ & $4$ & $5$ & $6$ & $7$ & $8$ & $9$\\
      \hline\hline
      $0$ & 155 & 214 & 228 & 232 & 323 & 396 & 570 & 641 & 715 & 796\\
      $1$ & 827 & 872 & 875 & 1048 & 1125 & 1158 & 1180
    \end{tabular}
  \end{center}
\item The known numbers $\gtz{m+1}{2,\dots,2,8}$, with $m = a \cdot 10 + b$:
    \begin{center}
    \begin{tabular}[c]{c||*{10}{c}}
      $\begin{array}[r]{c@{\hspace{0.5em}}r}
        & b\\[-1.9ex]
        a
      \end{array}$ & $0$ & $1$ & $2$ & $3$ & $4$ & $5$ & $6$ & $7$ & $8$ & $9$\\
      \hline\hline
      $0$ & 263 & 349 & 665 & 789 & 1323 & 1428 & 1447 & 1473 & 1555 & 1801\\
      $1$ & 1881 & 1935 & 1979 & 2117
    \end{tabular}
  \end{center}
\end{enumerate}

\end{appendix}


\begin{thebibliography}{10}

\bibitem{Ahmed2009vdW}
Tanbir Ahmed.
\newblock Some new van der {W}aerden numbers and some van der {W}aerden-type
  numbers.
\newblock {\em INTEGERS: Electronic Journal of Combinatorial Number Theory},
  9:65--76, 2009.

\bibitem{Ahmed2010vdW}
Tanbir Ahmed.
\newblock Two new van der {W}aerden numbers: w(2;3,17) and w(2;3,18).
\newblock To appear in INTEGERS: Electronic Journal of Combinatorial Number
  Theory, 2010.

\bibitem{2008HandbuchSAT}
Armin Biere, Marijn~J.H. Heule, Hans van Maaren, and Toby Walsh, editors.
\newblock {\em Handbook of Satisfiability}, volume 185 of {\em Frontiers in
  Artificial Intelligence and Applications}.
\newblock IOS Press, February 2009.

\bibitem{BraunsteinMezardZecchina2005SurveyProp}
A.~Braunstein, M.~M{\'{e}}zard, and R.~Zecchina.
\newblock Survey propagation: An algorithm for satisfiability.
\newblock {\em Random Structures and Algorithms}, 27(2):201--226, March 2005.

\bibitem{BrownLandmanRobertson2008WaerdenBounds}
Tom Brown, Bruce~M. Landman, and Aaron Robertson.
\newblock Bounds on some van der {W}aerden numbers.
\newblock {\em Journal of Combinatorial Theory, Series A}, 115:1304--1309,
  2008.

\bibitem{DransfieldLiuMarekTruszcynski2004VanderWaerden}
Michael~R. Dransfield, Lengning Liu, Victor~W. Marek, and Miroslaw
  Truszczy{\'{n}}ski.
\newblock Satisfiability and computing van der {W}aerden numbers.
\newblock {\em The Electronic Journal of Combinatorics}, 11(\#{}R41), 2004.

\bibitem{EenSoerensson2003Minisat}
Niklas E{\'{e}}n and Niklas S{\"{o}}rensson.
\newblock An extensible {SAT}-solver.
\newblock In Enrico Giunchiglia and Armando Tacchella, editors, {\em Theory and
  Applications of Satisfiability Testing 2003}, volume 2919 of {\em Lecture
  Notes in Computer Science}, pages 502--518, Berlin, 2004. Springer.
\newblock ISBN 3-540-20851-8.

\bibitem{ErdoesTuran1936vdW}
Paul Erd{\"{o}}s and P.~Tur{\'{a}}n.
\newblock On some sequences of integers.
\newblock {\em Journal of the London Mathematical Society}, 11:261--264, 1936.

\bibitem{GreenTao2005Primes}
Ben Green and Terence Tao.
\newblock The primes contain arbitrarily long arithmetic progressions.
\newblock {\em Annals of Mathematics}, 167(2):481--547, 2008.

\bibitem{GrosswaldHagis1979ProgressionsPrimes}
Emil Grosswald and Jr. Peter~Hagis.
\newblock Arithmetic progressions consisting only of primes.
\newblock {\em Mathematics of Computation}, 33(148):1343--1352, October 1979.

\bibitem{HerwigHeuleLambalgenMaaren2005VanderWaerden}
P.R. Herwig, M.J.H. Heule, P.M. van Lambalgen, and H.~van Maaren.
\newblock A new method to construct lower bounds for van der {W}aerden numbers.
\newblock {\em The Electronic Journal of Combinatorics}, 14(\#{}R6), 2007.

\bibitem{Heule2008PhD}
Marijn~J.H. Heule.
\newblock {\em SMART solving: Tools and techniques for satisfiability solvers}.
\newblock PhD thesis, Technische Universiteit Delft, 2008.
\newblock ISBN 978-90-9022877-8.

\bibitem{Kullmann2007HandbuchMU}
Hans {Kleine B{\"{u}}ning} and Oliver Kullmann.
\newblock Minimal unsatisfiability and autarkies.
\newblock In Biere et~al. \cite{2008HandbuchSAT}, chapter~11, pages 339--401.

\bibitem{KourilPaulW26}
Michal Kouril and Jerome~L. Paul.
\newblock The van der {W}aerden number ${W}(2,6)$ is $1132$.
\newblock {\em Experimental Mathematics}, 17(1):53--61, 2008.

\bibitem{Ku2002h}
Oliver Kullmann.
\newblock Investigating the behaviour of a {SAT} solver on random formulas.
\newblock Technical Report CSR 23-2002, Swansea University, Computer Science
  Report Series (available from
  \url{http://www-compsci.swan.ac.uk/reports/2002.html}), October 2002.
\newblock 119 pages.

\bibitem{Kullmann2009OKlibrary}
Oliver Kullmann.
\newblock The {\OKlibrary}: Introducing a "holistic" research platform for
  (generalised) {SAT} solving.
\newblock {\em Studies in Logic}, 2(1):20--53, 2009.

\bibitem{Kullmann2007ClausalFormECCC2}
Oliver Kullmann.
\newblock Constraint satisfaction problems in clausal form: Autarkies and
  minimal unsatisfiability.
\newblock Technical Report TR 07-055, version 02, Electronic Colloquium on
  Computational Complexity (ECCC), January 2009.

\bibitem{Kullmann2007HandbuchTau}
Oliver Kullmann.
\newblock Fundaments of branching heuristics.
\newblock In Biere et~al. \cite{2008HandbuchSAT}, chapter~7, pages 205--244.

\bibitem{Kullmann2007ClausalFormZI}
Oliver Kullmann.
\newblock Constraint satisfaction problems in clausal form {I}: Autarkies and
  deficiency.
\newblock {\em Fundamenta Informaticae}, 2010.
\newblock To appear.

\bibitem{Kullmann2007ClausalFormZII}
Oliver Kullmann.
\newblock Constraint satisfaction problems in clausal form {II}: Minimal
  unsatisfiability and conflict structure.
\newblock {\em Fundamenta Informaticae}, 2010.
\newblock To appear.

\bibitem{Kullmann2010GreenTao}
Oliver Kullmann.
\newblock Green-{T}ao numbers and {SAT}.
\newblock In Ofer Strichman and Stefan Szeider, editors, {\em Theory and
  Applications of Satisfiability Testing - SAT 2010}, Lecture Notes in Computer
  Science. Springer, 2010.

\bibitem{LandmanRobertsonCulver2005vanderWaerden}
Bruce Landman, Aaron Robertson, and Clay Culver.
\newblock Some new exact van der {Waerden} numbers.
\newblock {\em INTEGERS: Electronic Journal of Combinatorial Number Theory},
  5(2):1--11, 2005.
\newblock \#{}A10.

\bibitem{LandmanRobertson2003ArithmeticProgressions}
Bruce~M. Landman and Aaron Robertson.
\newblock {\em Ramsey Theory on the Integers}, volume~24 of {\em Student
  mathematical library}.
\newblock American Mathematical Society, 2003.
\newblock ISBN 0-8218-3199-2.

\bibitem{Li1999Satz}
Chu~Min Li.
\newblock A constraint-based approach to narrow search trees for
  satisfiability.
\newblock {\em Information Processing Letters}, 71(2):75--80, 1999.

\bibitem{Pre09HBSAT}
Steven Prestwich.
\newblock {CNF} encodings.
\newblock In Biere et~al. \cite{2008HandbuchSAT}, chapter~2, pages 75--97.

\bibitem{Roth1953vdW}
K.F. Roth.
\newblock On certain sets of integers.
\newblock {\em Journal of the London Mathematical Society}, 28:245--252, 1953.

\bibitem{Wagstaff1967ArithProg}
Jr. Samuel S.~Wagstaff.
\newblock On sequences of integers with no 4, or no 5 numbers in arithmetical
  progression.
\newblock {\em Mathematics of Computation}, 21(100):695--699, October 1967.

\bibitem{Wagstaff1972ArithProg}
Jr. Samuel S.~Wagstaff.
\newblock On k-free sequences of integers.
\newblock {\em Mathematics of Computation}, 26(119):767--771, July 1972.

\bibitem{Szemeredi1969AP}
E.~Szemer{\'{e}}di.
\newblock On sets of integers containing no four elements in arithmetic
  progression.
\newblock {\em Acta Mathematica Academiae Scientiarum Hungaricae}, 20:89--104,
  1969.

\bibitem{Szemeredi1975AP}
E.~Szemer{\'{e}}di.
\newblock On sets of integers containing no $k$ elements in arithmetic
  progression.
\newblock {\em Acta Arithmetica}, 27:299--345, 1975.

\bibitem{TompkinsHoos2004Ubcsat}
Dave~A.D. Tompkins and Holger~H. Hoos.
\newblock {UBCSAT}: An implementation and experimentation environment for {SLS}
  algorithms for {SAT} and {MAX-SAT}.
\newblock In Holger~H. Hoos and David~G. Mitchell, editors, {\em Theory and
  Applications of Satisfiability Testing 2004}, volume 3542 of {\em Lecture
  Notes in Computer Science}, pages 306--320, Berlin, 2005. Springer.
\newblock ISBN 3-540-27829-X.

\bibitem{vanderWaerden1927Baudet}
B.L. van~der Waerden.
\newblock Beweis einer {B}audetschen {V}ermutung.
\newblock {\em Nieuw Archief voor Wiskunde}, 15:212--216, 1927.

\bibitem{Zha09HBSAT}
Hantao Zhang.
\newblock Combinatorial designs by {SAT} solvers.
\newblock In Biere et~al. \cite{2008HandbuchSAT}, chapter~17, pages 533--568.

\end{thebibliography}
\end{document}